\documentclass[journal]{IEEEtran}
\usepackage[final]{graphicx}
\usepackage{subfigure}
\usepackage{float}
\usepackage{amsmath}
\usepackage{cases}
\usepackage{color}
\usepackage{colortbl}
\usepackage{algorithm}
\usepackage{algpseudocode}
\usepackage{amsmath}
\usepackage{amssymb}
\usepackage{booktabs}
\usepackage{setspace}
\usepackage{array}
\usepackage{rotating}
\usepackage[normalem]{ulem}
\usepackage{bbding}

\renewcommand{\normalsize}{\fontsize{10}{12}\selectfont}

\newtheorem{lemma}{\noindent \bf Lemma}
\newtheorem{theorem}{ \bf Theorem}
\newenvironment{proof}{{ \noindent \it Proof.}}{\hfill $\blacksquare$}

\makeatletter
\renewcommand{\maketag@@@}[1]{\hbox{\m@th\normalsize\normalfont#1}}%
\makeatother
\usepackage{stfloats}
\usepackage{cite}
\usepackage{makecell}
\usepackage{multirow}
\usepackage{hyperref}

\ifCLASSINFOpdf
\else
\fi

\usepackage{caption}
\usepackage{mathtools}
\usepackage{etoolbox} 

\begin{document}
\title{Carrier Aggregation Enabled MIMO-OFDM Integrated Sensing and Communication}
\author{Haotian Liu,~\IEEEmembership{Graduate Student Member,~IEEE,}
Zhiqing Wei,~\IEEEmembership{Member,~IEEE,} 
Jinghui Piao,
Huici Wu,~\IEEEmembership{Member,~IEEE,}
Xingwang Li,~\IEEEmembership{Senior Member,~IEEE,}\\
Zhiyong Feng,~\IEEEmembership{Senior Member,~IEEE}

\thanks{Haotian Liu, Zhiqing Wei, Jinghui Piao, Huici Wu, Zhiyong Feng are with the Beijing University of Posts and Telecommunications, Beijing 100876, China (emails: \{haotian\_liu; weizhiqing; piaojinghui; dailywu; fengzy\}@bupt.edu.cn). Corresponding authors: \textit{Zhiqing Wei, Haotian Liu.}

Xingwang Li is with the Henan Polytechnic University, Jiaozuo 454000,
China (e-mail: lixingwangbupt@gmail.com).

}}

\maketitle

\begin{abstract}
In the evolution towards the forthcoming era of sixth-generation (6G) mobile communication systems characterized by ubiquitous intelligence, integrated sensing and communication (ISAC) is in a phase of burgeoning development. 
However, the capabilities of communication and sensing within single frequency band 
fall short of meeting the escalating demands.
To this end,
this paper introduces a carrier aggregation (CA)-enabled multi-input multi-output orthogonal frequency division multiplexing (MIMO-OFDM) ISAC system fusing the sensing data on high and low-frequency bands by symbol-level fusion for ultimate communication experience and high-accuracy sensing.
The challenges in sensing signal processing introduced by CA include the initial phase misalignment of the echo signals on high and low-frequency bands due to attenuation and radar cross section, and the fusion of the sensing data on high and low-frequency bands with different physical-layer parameters. 
To this end, the sensing signal processing is decomposed into two stages. In the first stage, the problem of initial phase misalignment of the echo signals on high and low-frequency bands is solved by the angle compensation, spatial filtering and cyclic cross-correlation operations. In the second stage, this paper realizes symbol-level fusion of the sensing data on high and low-frequency bands through sensing vector rearrangement and cyclic prefix adjustment operations, thereby obtaining high-precision sensing performance.
Then, the closed-form communication mutual information (MI) and sensing Cram\'er-Rao lower bound (CRLB) for the proposed ISAC system are derived to explore the theoretical performance bound with CA. 
Simulation results validate the feasibility and superiority of the proposed ISAC system. 

\end{abstract}

\IEEEpeerreviewmaketitle

\section{Introduction}
The next-generation mobile communication system is anticipated to offer a plethora of multi-dimensional services extending beyond mere communication, such as smart transportation, smart factories, and digital twins~\cite{deng2023dynamic,zhuang2024,wei2024deep,zhou1}, which have raised the demands for high-speed communication and high-accuracy sensing. To this end, integrated sensing and communication (ISAC) is considered to provide the high-speed communication experience and high-accuracy sensing~\cite{wang2023deep,du2024nested,liu2018toward,wei2024deep}.

Nevertheless, the limitations imposed by sensing under single frequency band have spurred the study of ISAC over multiple frequency bands with the technique of carrier aggregation (CA)~\cite{wei2023carrier,wei2024deep,elbir2024curse,adamu2024analysis}. 
CA-enabled ISAC systems 
have the following potential advantages.
\begin{itemize}
    \item \textbf{Communication:} CA enhances bandwidth utilization and boosts communication transmission data rate while supporting various wireless access technologies, thereby improving overall communication performance~\cite{PedersenCA,wei2023carrier}.
    \item \textbf{Sensing:} CA empowers sensing systems with expanded bandwidth, which facilitates high-resolution sensing and anti-noise performance~\cite{wei2023carrier,wei2024deep}.
\end{itemize}
Therefore, it is reasonable to utilize CA to empower ISAC mobile communication system, which is expected to bring expeditious communication and precision sensing.


\begin{table*}[!ht]
\centering
\caption{A summary for the related work, with the abbreviations OFDM: orthogonal frequency division multiplexing, DFT: discrete Fourier transmit, CS: compressed sensing, NC: non-continuous, BWE: bandwidth extrapolation.}
\label{tab1}
\resizebox{0.8\textwidth}{!}{%
\renewcommand{\arraystretch}{1}
\begin{tabular}{|c|c|c|c|}
\hline
\textbf{Types} & \multicolumn{1}{c|}{\textbf{\begin{tabular}[c]{@{}c@{}}related\\ work\end{tabular}}} & \multicolumn{1}{c|}{\textbf{Innovation}} & \multicolumn{1}{c|}{\textbf{Shortcoming}} \\
\hline
\multirow{8}{*}{\textbf{\begin{tabular}[c]{@{}c@{}}Intra-band \\ CA\end{tabular}}} & \cite{pfeffer2015stepped}  & {\begin{tabular}[c]{@{}c@{}}Introduce the concept of stepped-carrier OFDM \\ radar to achieve high-resolution sensing\end{tabular}} & Only suitable for low speed targets  \\ \cline{2-4} &
\cite{schweizer2017stepped} & {\begin{tabular}[c]{@{}c@{}}Propose a novel DFT method to address the \\shortcoming in \cite{pfeffer2015stepped}\end{tabular}} & Poor anti-noise performance \\  \cline{2-4} &
\cite{huang2017nc} & {\begin{tabular}[c]{@{}c@{}}An ISAC system under NC fragmented spectrum bands \\is proposed to enhance spectrum utilization\end{tabular}}   &  The degradation of Fourier sidelobes \\ \cline{2-4} &
\cite{liu2023isac} & {\begin{tabular}[c]{@{}c@{}}A CS-based sensing method is proposed \\ to overcome the shortcoming in \cite{huang2017nc}\end{tabular}} & High computational complexity \\ \hline     
\multirow{5}{*}{\textbf{\begin{tabular}[c]{@{}c@{}}Inter-band\\ CA\end{tabular}}} &
\cite{cuomo1992bandwidth} &
{\begin{tabular}[c]{@{}c@{}}A BWE method is proposed to recovery \\the data on missing spectrum bands\end{tabular}} & The resolution performance is limited \\ \cline{2-4} &
\cite{suwa2004bandwidth} &
{\begin{tabular}[c]{@{}c@{}}Extend the BWE method to polarimetric radar \\ results in a higher enhancement in resolution\end{tabular}} & Unsuitable for OFDM signal \\ \cline{2-4} &
\cite{wei2023carrier} &{\begin{tabular}[c]{@{}c@{}} Joint high and low-frequency bands OFDM signal\\ for high-accuracy sensing  \end{tabular}}& Unsuitable for multiple antennas system  \\  \hline          
\end{tabular} }
\end{table*}

The researches on CA-enabled ISAC systems are classified into two types: intra-band CA and inter-band CA. The main challenges of CA-enabled ISAC system lie in signal design and processing. A summary of the related work is shown in Table \ref{tab1} and the details are as follows.
\begin{itemize}
    \item \textit{Intra-band CA:}
    Two primary ISAC signals are designed for intra-band CA: stepped-carrier OFDM ISAC signals and non-continuous (NC) fragmented OFDM ISAC signals.
    In terms of stepped-carrier OFDM ISAC signals, 
    Pfeffer~\textit{et al.} in \cite{pfeffer2015stepped} proposed the concept of stepped-carrier OFDM signals, which assigns different carrier frequencies to each resource block group during the continuous OFDM symbol periods, expanding the overall bandwidth and enhancing the sensing performance. However, it fails to consider the Doppler frequency shift induced by rapidly moving targets.
    To this end,
    Schweizer~\textit{et al.} in~\cite{schweizer2017stepped} explored a modified DFT method to mitigate the phase error caused by the high mobility of targets in stepped-carrier OFDM scheme, which obtains a high-accuracy velocity estimation. 
    In terms of NC fragmented OFDM ISAC signals,
    Huang~\textit{et al.} in~\cite{huang2017nc} introduced a NC fragmented OFDM ISAC system, which can flexibly allocate spectrum bands and adapt to dynamic spectrum environment. However, the missing spectrum bands deteriorate Fourier sidelobes in target sensing.
    To address this problem,
    Liu~\textit{et al.} in~\cite{liu2023isac} proposed a joint compressed sensing (CS) and machine learning sensing method for fragmented spectrum bands to achieve high-accuracy estimations of range and velocity with low sidelobes.
    \item \textit{Inter-band CA:}
    Multi-band radar stands as a representative example in inter-band CA.
    Cuomo in~\cite{cuomo1992bandwidth} discussed a NC multi-band coherent radar and proposed a bandwidth extrapolation (BWE) method to recover missing spectrum bands, thereby enhancing the range resolution of radar.
    Suwa~\textit{et al.} in~\cite{suwa2004bandwidth} extended the BWE method to polarimetric radar, which achieves high resolution of range estimation. 
\end{itemize}
    
In summary, intra-band and inter-band CA enhance the efficient utilization of fragmented frequency bands, leading to performance gains in both communication and sensing. Currently, the fragmented frequency bands in mobile communication system are mostly dispersed, indicating broader applications in mobile communication system under inter-band CA.
However, the investigation into inter-band CA in ISAC mobile communication systems is notably scarce. In our previous work, we have investigated the CA-enabled OFDM ISAC system, corroborating the effectiveness of CA through the validation~\cite{wei2023carrier}. However, our initial exploration, although valuable, remains somewhat rudimentary and idealized, overlooking the typical MIMO-OFDM signal in mobile communication systems and the challenges of sensing signal processing caused by practical channel environment. Therefore, it is particularly crucial to further study the feasibility and practicability of inter-band CA-enabled ISAC signal processing.

Hence, this paper investigates a CA-enabled MIMO-OFDM ISAC system where the fragmented high and low-frequency bands are shared for downlink (DL) communication and sensing potential targets. In addition, a CA-enabled MIMO-OFDM ISAC signal model and the corresponding sensing processing method are presented. Through theoretical derivations and numerical simulations, we validate the superiority of the proposed ISAC system over conventional ISAC systems. The main contributions of this paper are summarized as follows.
\begin{itemize}
    \item We present a CA-enabled MIMO-OFDM ISAC system aiming at achieving unparalleled communication experience and high-accuracy sensing. By introducing CA and MIMO, the communication capacity of ISAC system is significantly enhanced, as well as improving the accuracy and anti-noise capability of sensing. The simulation results validate the feasibility and efficiency of the proposed ISAC system, demonstrating its advantages over conventional ISAC systems.
    \item In terms of sensing processing, we delve into the data-level and symbol-level fusion of the sensing data on high and low-frequency bands. Data-level fusion refers to the fusion of multiple target information estimated by sensing data, while symbol-level fusion refers to the fusion of received symbols with phase information~\cite{wei2023symbol,wei2024integrated}. In the symbol-level fusion of sensing data, two challenges manifest:
    1) The initial phase misalignment of the echo signals on high and low-frequency bands due to attenuation and radar cross section (RCS); 2) The fusion of the sensing data on high and low-frequency bands with different physical-layer parameters. To this end, this paper decomposes the sensing signal processing into signal preprocessing stage and sensing information fusion stage. 
    The signal preprocessing stage fully leverage the sensing data from multiple antennas while aligning the initial phases of the echo signals on high and low-frequency bands using the cyclic cross-correlation (CCC) method, as detailed in Section \ref{se3-A}. In the sensing information fusion stage, this paper realizes the symbol-level fusion of the sensing data on high and low-frequency bands, as well as the range and velocity estimations of targets by rearranging the high and low-frequency feature vectors and using the premise of cyclic prefix (CP) dynamic adjustment, as detailed in Section \ref{sec3-B}. 
    \item The theoretical boundaries of the proposed CA-enabled MIMO-OFDM ISAC signal are derived. Specifically, mutual information (MI) is adopted as the performance metric for communication, while the Cram\'er-Rao lower bound (CRLB) is utilized as a typical performance metric of sensing. Based on the proposed ISAC system, we derive the closed-form communication MI under the frequency-selective fading MIMO-OFDM channel model and sensing CRLBs for range and velocity estimations under the MIMO-OFDM sensing signal model. 
\end{itemize}

The rest of the paper is organized as follows. Section \ref{se2} introduces the system model, including the communication model and sensing model. In Section \ref{se3}, a novel CA-enabled MIMO-OFDM ISAC signal processing method is proposed. Section \ref{se4} derives the communication MI and sensing CRLB of the proposed CA-enabled MIMO-OFDM ISAC signal. The simulation results are provided in Section \ref{se5}, while the conclusions are outlined in Section \ref{se6}.

\textit{Notations:} $\{\cdot\}$ typically stands for a set of various index values. 
Black bold letters represent matrices or vectors.
$\mathbb{C}$ and $\mathbb{R}$ denote the set of complex and real numbers, respectively. 
$\left[\cdot\right]^{\text{T}}$, $\left[\cdot\right]^{\text{H}}$, $\left[\cdot\right]^{-1}$, and $\left(\cdot\right)^{\text{*}}$ stand for the transpose operator, conjugate transpose operator, inverse operator, and conjugate operator, respectively. 
For a complex-valued vector $\mathbf{u}$, $\text{diag}\left(\mathbf{u}\right)$ stands for a diagonal matrix whose diagonal elements are given by the elements of $\mathbf{u}$. $\text{det}\left(\mathbf{u}\right)$ and $\text{tr}\left(\mathbf{u}\right)$ are the determinant and trace of $\mathbf{u}$, respectively. $E\left(\cdot\right)$ stands for the expectation operator.
A complex Gaussian random variable $\mathbf{u}$ with mean $\mu_u$ and variance $\sigma_u^2$ is denoted by $\mathbf{u} \sim \mathcal{CN}\left(\mu_u,\sigma_u^2\right)$. 
$\odot$ is the Hadamard product.

\section{System Model} \label{se2}
We consider a CA-enabled MIMO-OFDM ISAC base station (BS) for DL communication and sensing potential targets as shown in Fig.~\ref{fig1}. 
The ISAC BS is equipped with $N_{\text{T}}+N_{\text{R}}$ antennas, where $N_{\text{T}}$ transmit antennas in transmitter (Tx) simultaneously serve $U$ user equipments (UEs) with $N_\text{U}$ antennas and detect $I$ potential targets. $N_{\text{R}}$ receive antennas in receiver (Rx) are equipped to receive the echo signals reflected by targets~\cite{hua2022integrated}.
The spacing of antennas in BS is denoted by $d_\text{r}$~\cite{liu2018}. The ISAC BS can transmit DL data information while simultaneously emitting sensing signals for target sensing. Data echo signals are used for UE sensing, and sensing echo signals for target sensing. Since the sensing processing methods are identical, this paper focuses on the processing of sensing echo signals.
\begin{figure}[!ht]   
    \centering    \includegraphics[width=0.42\textwidth]{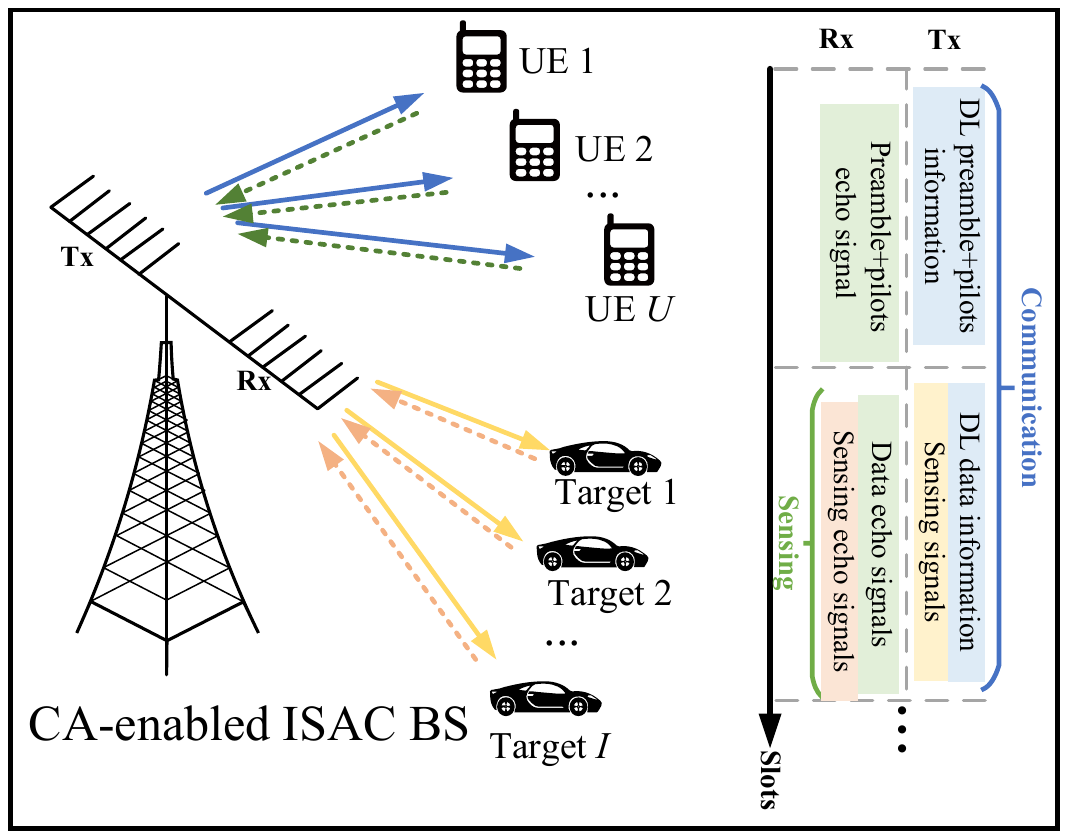}
    \caption{CA-enabled MIMO-OFDM ISAC system.}
    \label{fig1}
\end{figure}

\subsection{CA-enabled MIMO-OFDM ISAC Signal}
For the $b \in {\mathcal{B}}=\{1,2\}$-th carrier component (CC) of the ISAC BS, a total number of $N_b$ subcarriers and $M_b$ OFDM symbols are used for DL communication and sensing, where $n \in \{0, 1, \cdots ,N_b-1\}$ and $m \in \{0, 1, \cdots ,M_b-1\}$ denote the indices of subcarriers and OFDM symbols, respectively. 
For the $b$-th CC, the symbol data matrix~\cite{Hu2022MIMO-OFDM}
\begin{equation} \label{eq1}
    \mathbf{S}^b=\left[\mathbf{S}_0^b, \mathbf{S}_1^b, \cdots, \mathbf{S}_{N_b-1}^b\right]^\text{T} \in \mathbb{C}^{UN_b \times M_b}
\end{equation}
represents the data for all UEs on $N_b$ subcarriers during $M_b$ OFDM symbol times, where 
$\mathbf{S}_n^b  =\left[\mathbf{s}_{n,1}^b, \mathbf{s}_{n,2}^b, \cdots, \mathbf{s}_{n,U}^b\right]\in
\mathbb{C}^{M_b \times U}$ indicates the data for all UEs on the $n$-th subcarrier during $M_b$ symbol times with $\mathbf{s}_{n,u}^b \in \mathbb{C}^{M_b \times 1}$ being the specific symbol data vector for the $u \in \{1,2,\cdots, U\}$-th UE~\cite{Hu2022MIMO-OFDM}. 
There is a trade-off between the randomness of the communication signal and the certainty of the sensing signal. To guarantee the balanced ISAC performance, this paper adopts the low complexity data independent precoding (DIP) scheme proposed in~\cite{lu2023random}, which can be applied to various ISAC systems. Specifically, an ISAC precoding matrix~\cite{Hu2022MIMO-OFDM,lu2023random} 
\begin{equation}  \label{eq2}  \mathbf{W}^b=\mathrm{diag}\left(\mathbf{W}_{0}^b,\mathbf{W}_{1}^b,\cdots,\mathbf{W}_{N_b-1}^b\right)\in\mathbb{C}^{N_{b}N_{\text{T}} \times N_{b}U}
\end{equation}
is designed for all UEs on $N_b$ subcarriers. $\mathbf{W}_n^b=\left[\mathbf{w}_{n,1}^b,\mathbf{w}_{n,2}^b,\cdots,\mathbf{w}_{n,U}^b\right]\in\mathbb{C}^{N_{\text{T}}\times U}$ denotes the precoding matrix for all UEs on the $n$-th subcarriers with $\mathbf{w}_{n,u}^b \in \mathbb{C}^{N_\text{T} \times 1}$ being the specific precoding vector for the $u$-th UE~\cite{Hu2022MIMO-OFDM}. Therefore, the transmit precoded data on $N_b$ subcarriers during $M_b$ symbol times is~\cite{Hu2022MIMO-OFDM} 
\begin{equation} \label{eq3}  
\begin{aligned}
\mathbf{X}^b&=\mathbf{W}^b\mathbf{S}^b \\ & =\left[(\mathbf{X}_{0}^b)^{\text{T}},(\mathbf{X}_{1}^b)^{\text{T}},\cdots,(\mathbf{X}_{N_b-1}^b)^{\text{T}}\right]^{\text{T}}\in\mathbb{C}^{N_bN_{\text{T}}\times M_b}, 
\end{aligned}
\end{equation}
where $\mathbf{X}_{n}^b=\mathbf{W}^b(\mathbf{S}_n^b)^\text{T} \in \mathbb{C}^{N_\text{T} \times M_b}$ is the data on the $n$-th subcarrier during $M_b$ OFDM symbol times. 
The transmit precoded data is transformed into time-domain signal by undergoing an inverse discrete Fourier transform (IDFT). 

Therefore, the baseband time-domain transmit signal by the $k$-th transmit antenna on the $n'$-th subcarrier during the $m$-th OFDM symbol time in the $b$-th CC is~\cite{Parag2018}
\begin{equation} \label{eq4}
x^b\left(k,n',m\right)=\frac1{N_{b}}\sum_{n=0}^{N_{b}-1}\mathbf{x}_k^b\left(n,m\right)e^{j2\pi \frac {n'n}{N_b}},
\end{equation}
where $\mathbf{x}_k^b\left(n,m\right)$ denotes the $\left(n,m\right)$-th element of $\mathbf{x}_k^b=[\mathbf{X}^b]_{(k-1)N_b+1:kN_b,:}$. $n' \in \{0,1,\cdots,N_b-1\}$ represents the index of subcarriers in time-domain and $k \in \{0,1,\cdots,N_\text{T}-1\}$ represents the index of transmit antenna.
A CP with size $N_\text{s}^b$ greater than the maximum delay caused by target reflection is added to avoid inter-symbol interference. 
After digital-to-analog conversion (DAC), the analog signals of the $1$-st and $2$-nd CCs are up-converted using different local oscillators (LOs), mixed, and transmitted via a linear high power amplifier (HPA), which avoids the nonlinear effects. 

Overall, the CA-enabled MIMO-OFDM ISAC signal by the $k$-th transmit antenna on the $n$-th subcarrier during the $m$-th OFDM symbol time is expressed as
\begin{equation}  \label{eq5}
{\fontsize{8}{8}\begin{aligned}
x(t)_{k}=\sum_{b=1}^{\mathcal{B}}\sum_{m=0}^{M_b-1}\sum_{n=0}^{N_b-1}\mathbf{x}_k^b\left(n,m\right)e^{j2\pi t(f_\mathrm{C}^b+n\Delta f^b)}\text{rect}\left(\frac{t-mT^b}{T^b}\right),
\end{aligned}}
\end{equation}
where $f_\mathrm{C}^b$ denotes the carrier frequency of the $b$-th CC; $\Delta f^b=1/T_\text{ofdm}^b$ is the subcarrier spacing of the $b$-th CC, where $T_\text{ofdm}^b$ is the elementary symbol duration~\cite{Parag2018}; $T^b=T_{\text{ofdm}}^b+T_\text{s}^b$ is the total symbol duration with $T_\text{s}^b=(N_\text{s}^bT_{\text{ofdm}}^b)/N_b$ being an adjustable CP duration.

\subsection{Communication Model}
In the communication channel model, since the high and low-frequency bands are processed independently without fusion, the differences in attenuations caused by high and low-frequency bands are not considered.

\subsubsection{Communication channel model}
We utilize the widely adopted multi-path channel model in MIMO channels, assuming the presence of $L$ delay paths in the communication environment. The index of receive antennas in UE is denoted by $z \in \{0,1,\cdots,N_\text{U}-1\}$ and
\begin{equation} \label{eq6}
{ \fontsize{9}{9}  \mathbf{H}_{u}^{l,b} = \left[\setlength{\arraycolsep}{1pt} 
       \begin{array}{cccc}
        h_{0,0}^{u,l,b} & h_{0,1}^{u,l,b} & \cdots & h_{0,N_\text{U}-1}^{u,l,b}  \\
        h_{1,0}^{u,l,b} & h_{1,1}^{u,l,b} & \cdots & h_{1,N_\text{U}-1}^{u,l,b}  \\
        \vdots & \vdots & \ddots & \vdots \\
        h_{N_\text{T}-1,0}^{u,l,b} & h_{N_\text{T}-1,1}^{u,l,b} &\cdots & h_{N_\text{T}-1,N_\text{U}-1}^{u,l,b}
       \end{array}
    \right]}\in \mathbb{C}^{N_\text{T} \times N_\text{U}}
\end{equation} 
is the channel matrix of the $l \in \{0,1,\cdots,L-1\}$-th delay path for the $u$-th UE with $h_{k,z}^{u,l,b}\sim \mathcal{CN}(0,1)$ being the corresponding channel element with independent identically distributed (i.i.d.) Rayleigh fading coefficient~\cite{goldsmith2003capacity,Hu2022MIMO-OFDM}. Therefore, the MIMO communication channel matrix in time-domain for the $u$-th UE on the $n$-th subcarrier during the $m$-th OFDM symbol time is expressed as~\cite{Wei2023MI,Parag2018}
\begin{equation}\label{eq7}
\mathbf{H}_{u,m}^n=\sum_{l=0}^{L-1}\mathbf{H}_u^{l,b}\delta{\left[m-l\right]}\in \mathbb{C}^{N_\text{T}\times N_\text{U}}.
\end{equation}

\subsubsection{Received communication signal}
At the DL preamble+pilots information period, pilots are used to obtain channel state information (CSI) estimation~\cite{cho2010mimo,chen2024downlink}. At the DL data information period, the $u$-th UE receives the DL communication signal, and undergoes down-conversion, analog-to-digital conversion (ADC), CP removal, DFT, and other operations to get a received data $\mathbf{Y}_{u}^\mathrm{C}$, which is expressed as
\begin{equation} \label{eq8} \mathbf{Y}_{u}^\mathrm{C}=\sum_{b=1}^{\mathcal{B}}\mathbf{X}_\mathrm{C}^b\mathbf{H}_\mathrm{C}^b+\mathbf{W}_{u}^\mathrm{C},
\end{equation}
where $\mathbf{H}_\mathrm{C}^b =\left[(\widetilde{\mathbf{H}}_{u,m}^0)^\text{T}, \cdots, (\widetilde{\mathbf{H}}_{u,m}^{N_b-1})^\text{T}\right] \in \mathbb{C}^{N_bN_\text{T}\times N_\text{U}}$ with $\widetilde{\mathbf{H}}_{u,m}^n=\sum_{l=0}^{L-1}\widetilde{\mathbf{H}}_u^{l,b} e^{\frac{-j2 \pi nl}{N_b}} \in \mathbb{C}^{N_\text{T}\times N_\text{U}}$, while $\widetilde{\mathbf{H}}_u^{l,b}$ is the estimated CSI matrix; $\mathbf{X}_\mathrm{C}^b=\text{diag}\left((\mathbf{X}_0^b)^\text{T}, \cdots, (\mathbf{X}_{N_b-1}^b)^\text{T}\right) \in \mathbb{C}^{N_bM_b\times N_bN_\text{T}}$ is the transmit data; $\mathbf{W}_{u}^\mathrm{C}\sim \mathcal{CN}\left(0,\sigma_\mathrm{C}^2\right)$ is additive Gaussian white noise (AWGN) of dual-CC communication channel.

Since linear HPA is considered in this paper, it is proposed that the communication processing of each CC of the proposed ISAC system is similar to that of the traditional MIMO-OFDM communication system, and is independent for each UE~\cite{singhal2015analysis,Parag2018}. In terms of communication performance in CA-enabled MIMO-OFDM ISAC system,  the closed-form communication MI is derived in Section \ref{se4-A}.

\subsection{Sensing Model}
To better reflect practical scenarios, we consider the factors that directly or indirectly influence the high and low-frequency aggregation in sensing signal processing, including phase alignment errors caused by attenuation differences, physical parameter variations, and noise variance discrepancies~\cite{3gpp.38.104}. To the best of our knowledge, this is the first work in CA-enabled ISAC research to consider phase alignment errors and noise variance discrepancies, making it more applicable to practical applications.

\subsubsection{Sensing channel model}
We assume that the radial velocity of the $i$-th potential target is $v_{i,0}$ and the range to BS is $r_{i,0}$. Meanwhile, the angle of arrive (AoA) and the angle of departure (AoD) between the $i$-th target and BS are identical and are denoted by $\theta_{i,\text{Rx}}=\theta_{i,\text{Tx}}$. The transmit and receive steering vectors are expressed as (\ref{eq9}) and (\ref{eq10}), respectively. 
\begin{align}
    \mathbf{a}_{\text{Rx}}(\theta_{i,\text{Rx}})&=\left[e^{j2\pi p (\frac{d_\text{r}}{\lambda^b})\sin(\theta_{i,\text{Rx}})}\right]^{\text{T}}|_{p=0,1,\cdots,N_\text{R}-1},\label{eq9}\\
    \mathbf{a}_{\text{Tx}}(\theta_{i,\text{Tx}})&=\left[e^{j2\pi k (\frac{d_\text{r}}{\lambda^b})\sin(\theta_{i,\text{Tx}})}\right]^{\text{T}}|_{k=0,1,\cdots,N_\text{T}-1},\label{eq10}
\end{align}
where $p$ denotes the index of the receive antenna in BS; $\lambda^b=c/f_\mathrm{C}^b$ is the wavelength with $c$ being the speed of light.

The sensing channel model on the $n$-th subcarrier during the $m$-th OFDM symbol time in the $b$-th CC is expressed as \cite{chen2024downlink,liu2024target}
\begin{equation} \label{eq11}
\mathbf{H}_{m,n}^b=\sum_{i=1}^{I}
 \left[\begin{array}{l}
 \kappa_{i,S}^b e^{j2\pi \left(f_{i,\text{s}}^b m T^b-n\Delta f^b \tau_{i,0}\right)} \\ \mathbf{a}_{\text{Rx}}(\theta_{i,\text{Rx}})\mathbf{a}_{\text{Tx}}^{\text{T}}(\theta_{i,\text{Tx}})\end{array}\right],
\end{equation}
where $\mathbf{H}_{m,n}^b \in \mathbb{C}^{N_{\text{R}}\times N_\text{T}}$ and $\kappa_{i,S}^b=\sqrt{\frac{(\lambda^b)^2}{(4\pi)^3r_{i,0}^4}}\beta_{i,S}^b$ is the attenuation between the $i$-th target and BS with $\beta_{i,S}^b \sim \mathcal{CN}(0,\sigma_{\beta_{i,S}^b}^2)$ being the 
$i$-th target's RCS \cite{guo2023}, which characterizes the difference in the initial phase of high and low-frequency bands;
$f_{i,\text{s}}^b=\frac{2f_\mathrm{C}^b v_{i,0}}{c}$ and $\tau_{i,0}=\frac{2r_{i,0}}{c}$ are the Doppler frequency shift and delay, respectively. 

\subsubsection{Received sensing signal}
\cite{Mateo} verifies that the high and low-frequency band signals have similar behavior, that is, they reach the receiver (Rx) of BS through the same line of sight path and target reflection. At the Rx of BS, the mixed sensing echo signals on high and low-frequency bands are separated through matched filtering. Therefore, the received echo signal on the $n$-th subcarrier during the $m$-th OFDM symbol time is denoted by
\begin{equation} \label{eq12}
\begin{aligned} \mathbf{y}_{m,n}^S=\sum_{b=1}^{\mathcal{B}}\mathbf{y}_{m,n}^{S,b} =\sum_{b=1}^{\mathcal{B}}\mathbf{H}_{m,n}^b\mathbf{W}_\text{Tx}^b\mathbf{d}_{m,n}^b+\mathbf{z}_{m,n}^{S,b},   
\end{aligned}
\end{equation}
where $\mathbf{y}_{m,n}^S \in \mathbb{C}^{N_\text{R}\times 1}$ and $\mathbf{y}_{m,n}^{S,b}$ denotes the echo signal in $b$-th CC; $\mathbf{W}_\text{Tx}^b\in\mathbb{C}^{N_\text{T}\times N_\text{T}}$ represents the transmit beamforming matrix used to point in the direction of interest, and $\mathbf{d}_{m,n}^b\in \mathbb{C}^{N_\text{T}\times1}$ is the dedicated sensing data with constant mode 1; $\mathbf{z}_{m,n}^{S,b}  \sim \mathcal{CN}(0,\sigma_{S,b}^2)\in \mathbb{C}^{N_\text{R}\times 1}$ is an AWGN vector, which characterizes the different noise variances of high and low-frequency bands.

\section{A Novel CA-enabled MIMO-OFDM ISAC Signal Processing Method} \label{se3}
\begin{figure*}
 \centering    \includegraphics[width=0.95\textwidth]{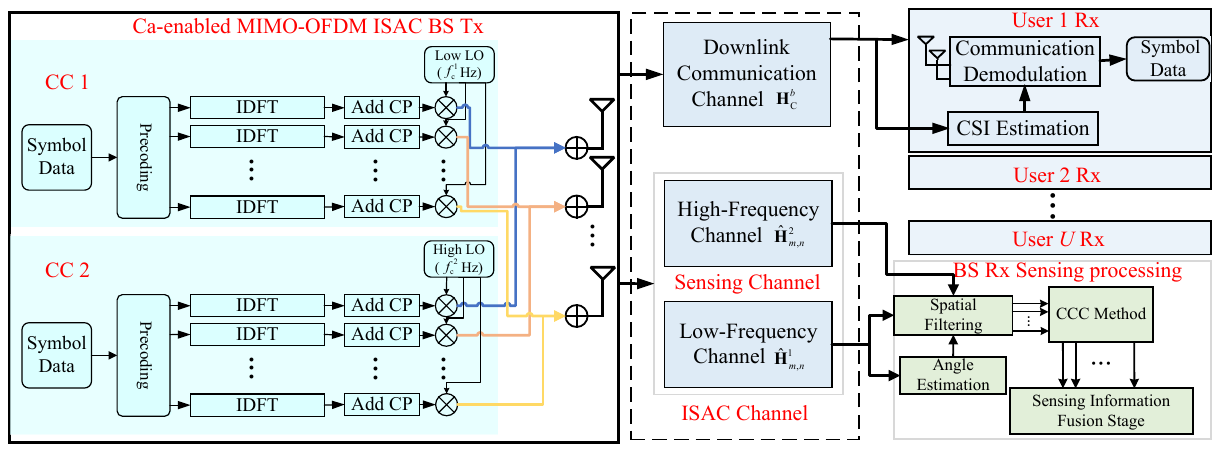}
    \caption{CA-enabled MIMO-OFDM ISAC signal processing.}
    \label{fig2}   
\end{figure*}
In this section, a novel CA-enabled MIMO-OFDM ISAC signal processing method is proposed, as illustrated in Fig. \ref{fig2}. Without loss of generality, we assume that the $1$-st CC is on the low-frequency band and the $2$-nd CC is on the high-frequency band. Through matched filtering at the Rx, the mixed sensing echo signals on the two CCs are separated.

There are two challenges in sensing signal processing within the proposed ISAC system: 1) The initial phase misalignment of the echo signals on high and low-frequency bands due to attenuation and RCS; 2) The fusion of the sensing data on high and low-frequency bands with different physical-layer parameters, such as subcarrier spacing and symbol duration. Therefore, we consider both data-level and symbol-level fusion methods of the sensing data on high and low-frequency bands.

The data-level fusion method is as follows: After the separation of mixed echo signals on high and low-frequency bands, the joint angle-range-velocity estimation method proposed in~\cite{xu2023bandwidth} is applied to the echo signal on high-frequency band to yield the estimated range set $\{r_{i,\text{data}}^1\}_{i=1,2,\cdots,I}$ and velocity set $\{v_{i,\text{data}}^1\}_{i=1,2,\cdots,I}$ of $I$ targets, which have a one-to-one relationship. Similarly, the estimated range set $\{r_{i,\text{data}}^2\}_{i=1,2,\cdots,I}$ and velocity set $\{v_{i,\text{data}}^2\}_{i=1,2,\cdots,I}$ of $I$ targets are obtained by processing the echo signal on low-frequency band. We combine the difference of accuracy and SNR to give the data-level fusion weights, and the variances of received signal on high and low-frequency bands can be obtained by some methods, such as maximum likelihood estimation (MLE)~\cite{Morelli}. Therefore, the final estimation for range of the $i$-th target is $\hat{r}_{i,\text{data}}=r_{i,\text{data}}^1+ \frac{\sigma_{S,1}^2 \frac{\Delta f^2 M_b}{\Delta f^1}}{\sigma_{S,1}^2 \frac{\Delta f^2 M_b}{\Delta f^1}+\sigma_{S,2}^2}(r_{i,\text{data}}^2-r_{i,\text{data}}^1)$, while the final estimation for velocity of the $i$-th target is $\hat{v}_{i,\text{data}}=v_{i,\text{data}}^1+\frac{\sigma_{S,1}^2\frac{\Delta f^2 N_b}{\Delta f^1}}{\sigma_{S,1}^2\frac{\Delta f^2 N_b}{\Delta f^1}+\sigma_{S,2}^2}(v_{i,\text{data}}^2-v_{i,\text{data}}^1)$~\cite{chen2024downlink}. However, the sensing accuracy of data-level fusion method is low.
Therefore, this paper proposes a two-stage symbol-level fusion method, including signal preprocessing stage and sensing information fusion stage, detailed in Section \ref{se3-A}.
\subsection{Signal Preprocessing Stage}\label{se3-A}
One primary purposes of sensing preprocessing is to align the initial phases of the echo signals on high and low-frequency bands.

\subsubsection{Angle estimation}\label{se3-A-1}
With identical AoA and AoD for two CCs and high signal-to-noise ratio (SNR) of the echo signal on low-frequency band, we utilize the echo signal on low-frequency band to estimate AoA.

Firstly, we remove the known transmit data $\mathbf{d}_{m,n}^1$ in $\mathbf{y}_{m,n}^{S,1}$ by right-multiplying the generalized inverse matrix $\mathbf{X}_{m,n}^\dagger = \left(\mathbf{d}_{m,n}^1\right)^\text{H}\left(\mathbf{d}_{m,n}^1 \left(\mathbf{d}_{m,n}^1\right)^\text{H}+\varrho\mathbf{I}\right)^{-1}$, and the echo signal on low-frequency band by all antennas is
\begin{equation}  \label{eq13} 
{\fontsize{9}{9}\begin{aligned}
\hat{\mathbf{H}}_{m,n}^1&=\mathbf{y}_{m,n}^{S,1}\mathbf{X}_{m,n}^\dagger \\ &
\approx \sum_{i=1}^I\left[\begin{array}{l}
\kappa_{i,S}^1 e^{j2\pi \left(f_{i,\text{s}}^1 m T^1-n\Delta f^1 \tau_{i,0}\right)} \\\mathbf{a}_{\text{Rx}}(\theta_{i,\text{Rx}})\mathbf{a}_{\text{Tx}}^{\text{T}}(\theta_{i,\text{Tx}})\mathbf{I}\end{array}\right] \mathbf{W}_\text{Tx}^b+\mathbf{z}_{m,n}^{S,1}\mathbf{X}_{m,n}^\dagger,
\end{aligned}}
\end{equation}
where $\hat{\mathbf{H}}_{m,n}^1 \in \mathbb{C}^{N_\text{R}\times N_\text{T}}$ and $\mathbf{I} \in \mathbb{C}^{N_\text{T}\times N_\text{T}}$ is the identity matrix. 
$\varrho$ denotes a regularization parameter introduced to improve the accuracy of the numerical recovery~\cite{wei2024integrated}.

Upon observing $\hat{\mathbf{H}}_{m,n}^1$, it is obvious that the AoA introduces 
a linear phase shift along the receive antenna elements. The AoA estimations can be obtained by multiple signal classification (MUSIC) method, estimating signal parameter via rotational invariance techniques method, and the advanced methods based on complex neural network~\cite{wei2024deep,Naoumi}. Due to its ability to achieve parameter estimation in non-contiguous signals in the space-time-frequency-code domains with high precision,  we take the MUSIC method as an example~\cite{schmidt1986multiple,chen2024downlink}, with the following steps.

\textit{Step 1:} Calculate the covariance matrix $\mathbf{R}_{\hat{\mathbf{H}}_{m,n}^1}$ of  $\hat{\mathbf{H}}_{m,n}^1$.
\begin{equation} \label{eq14}
\mathbf{R}_{\hat{\mathbf{H}}_{m,n}^1}=\frac{\hat{\mathbf{H}}_{m,n}^1\left[\hat{\mathbf{H}}_{m,n}^1\right]^{\text{H}}}{N_1M_1}\in \mathbb{C}^{N_\text{R}\times N_\text{R}}.
\end{equation}

\textit{Step 2:} $\mathbf{R}_{\hat{\mathbf{H}}_{m,n}^1}$ is performed an eigenvalue decomposition $\text{eig}(\cdot)$ to obtain 
\begin{equation} \label{eq15}
\text{eig}\left(\mathbf{R}_{\hat{\mathbf{H}}_{m,n}^1}\right)=\left[\mathbf{U}_\text{s},\mathbf{\Lambda}_\text{s}\right],
\end{equation}
where $\mathbf{\Lambda}_\text{s}$ denotes a diagonal matrix with descending order of eigenvalues. $\mathbf{U}_\text{s}$ is the orthogonal eigenmatrix.
The number of incident signals can be accurately ascertained by estimating the rank of the signal space from the differential vector derived from the eigenvalues of the signal matrix~\cite{chen2023multiple}, which are assumed to be the same as the number of targets.
Then, the noise subspace $\mathbf{U}_\text{n}=\left[\mathbf{U}_\text{s}\right]_{:,I+1:N_\text{R}}$ is obtained, which is used to calculate the spatial spectral function.

\textit{Step 3:} Generate a MUSIC spatial spectral vector 
\begin{equation} \label{eq16}
    f_{\text{music}}(\theta)=\frac{1}{\mathbf{a}(\theta)^\text{H}\mathbf{U}_\text{n}\mathbf{U}_\text{n}^{\text{H}}\mathbf{a}(\theta)},
\end{equation}
where $\theta \in \left(0, \pi\right]$ and
\begin{equation}\label{eq17}
    \mathbf{a}(\theta)=\left[e^{j2\pi p (\frac{d_\text{r}}{\lambda^1})\sin(\theta)}\right]^\text{T} |_{p=0,1,\cdots,N_\text{R}-1}.
\end{equation}

\textit{Step 4:} Search the $I$ peaks of $f_{\text{music}}$ to obtain the AoA estimations set $\{\hat{\theta}_{i',\text{Rx}}\}_{i'=1,2\cdots,I}$.

\subsubsection{Spatial filtering}\label{se3-A-2}
To mitigate target ghosts introduced by potential nonlinear processing, we employ spatial filtering operation to separate the echo signals that are blended from multiple targets.

Similar to (\ref{eq13}), after the known transmit data $\mathbf{d}_{m,n}^b$ is removed, the echo signal by all antennas on the $n$-th subcarrier during the $m$-th OFDM symbol time in the $b$-th CC is expressed as
\begin{equation}\label{eq18}
  \hat{\mathbf{H}}_{m,n}^{b}\approx \sum_{i=1}^I\left[\begin{array}{l}
\kappa_{i,S}^b e^{j2\pi \left(f_{i,\text{s}}^b m T^b-n\Delta f^b \tau_{i,0}\right)} \\\mathbf{a}_{\text{Rx}}(\theta_{i,\text{Rx}})\mathbf{a}_{\text{Tx}}^{\text{T}}(\theta_{i,\text{Tx}})\mathbf{I}\end{array}\right]\mathbf{W}_\text{Tx}^b+\mathbf{Z}_{m,n}^{S,b},
\end{equation}
where $\mathbf{Z}_{m,n}^{S,b}$ is a noise matrix. 

Similar to the principle of OFDM signal demodulation, we can separate the individual echo signal of each target from $\hat{\mathbf{H}}_{m,n}^{b}$ with the estimated AoAs $\{\hat{\theta}_{i',\text{Rx}}\}_{i'=1,2\cdots,I}$. The echo signal of the $i'$-th target by all antennas on the $n$-th subcarrier during the $m$-th OFDM symbol time in the $b$-th CC is expressed in (\ref{eq19}). When the target angle difference is greater than the angular resolution of the antenna, and the SNR is not poor, the magnitude of the \textbf{Term 1} is much greater than that of the \textbf{Term 2}. Fortunately, the spatial filtering operation can confer substantial SNR gains in space-domain.
\begin{figure*}
  \begin{equation}\label{eq19}
   \begin{aligned}
     \hat{\mathbf{H}}_{m,n}^{b,i'}&= \mathbf{a}_\text{Rx}^\text{H}(\hat{\theta}_{i',\text{Rx}})\hat{\mathbf{H}}_{m,n}^{b}(\mathbf{W}_\text{Tx}^b)^\text{H}\mathbf{a}_\text{Tx}^\text{*}(\hat{\theta}_{i',\text{Tx}})/(N_\text{R}N_\text{T}) \\ &
     \approx \underbrace{\kappa_{i',S}^b e^{j2\pi \left(f_{i',\text{s}}^b m T^b-n\Delta f^b \tau_{i',0}\right)}}_{\text{The}\ i' \text{-th target} \ (\textbf{Term 1})} \\ &
     \quad+ \underbrace{\mathbf{a}_\text{Rx}^\text{H}(\hat{\theta}_{i',\text{Rx}})\left[\sum_{i=1,i\neq i'}^{I}\left[\begin{array}{l}
\kappa_{i,S}^b e^{j2\pi \left(f_{i,\text{s}}^b m T^b-n\Delta f^b \tau_{i,0}\right)} \\\mathbf{a}_{\text{Rx}}(\theta_{i,\text{Rx}})\mathbf{a}_{\text{Tx}}^{\text{T}}(\theta_{i,\text{Tx}})\mathbf{I}\end{array}\right]\mathbf{W}_\text{Tx}^b+\mathbf{Z}_{m,n}^{S,b}\right](\mathbf{W}_\text{Tx}^b)^\text{H}\mathbf{a}_\text{Tx}^\text{*}(\hat{\theta}_{i',\text{Tx}})/(N_\text{R}N_\text{T}).}_{\text{Interference}\ (\textbf{Term 2})}
    \end{aligned}
\end{equation} 
{\noindent} \rule[-10pt]{18cm}{0.1em}
\end{figure*}

According to (\ref{eq19}), the echo signal of the $i'$-th target on $N_b$ subcarriers during $M_b$ OFDM symbol times in the $b$-th CC is rewritten in matrix form as
\begin{equation} \label{eq20}  \mathbf{D}_{i'}^{S,b}=\mathbf{S}_{i'}^{S,b}+\mathbf{Z}_{i'}^{S,b},
\end{equation}
where $\mathbf{Z}_{i'}^{S,b}$ is an interference matrix including the interference from other target and noise, and $\mathbf{S}_{i'}^{S,b} \in \mathbb{C}^{N_b \times M_b}$ denotes a delay-Doppler information matrix, expressed in (\ref{eq21}).
\begin{figure*}
\begin{equation}
    \centering
   \label{eq21}    \mathbf{S}_{i'}^{S,b}=\kappa_{i',S}^b\left[
    \begin{array}{cccc}
     1 & e^{j2\pi f_{i',\text{s}}^b T^b} & \cdots & e^{j2\pi (M_b-1) f_{i',\text{s}}^b  T^b} \\
      e^{-j2\pi \Delta f^b \tau_{i',0}} & e^{-j2\pi \Delta f^b \tau_{i',0}} e^{j2\pi f_{i',\text{s}}^b T^b} & \cdots &  e^{-j2\pi \Delta f^b \tau_{i',0}} e^{j2\pi (M_b-1) f_{i',\text{s}}^b T^b}   \\
     \vdots & \vdots &  \ddots & \vdots \\
      e^{-j2\pi (N_b-1) \Delta f^b \tau_{i',0}} & e^{-j2\pi (N_b-1) \Delta f^b \tau_{i',0}}e^{j2\pi f_{i',\text{s}}^b T^b} & \cdots & e^{-j2\pi (N_b-1) \Delta f^b \tau_{i',0}}e^{j2\pi (M_b-1) f_{i',\text{s}}^b T^b} \\
    \end{array} \right].
    \end{equation}
\end{figure*}

\subsubsection{CCC method}
For the $i'$-th target, observing (\ref{eq21}), the initial phases of $\mathbf{S}_{i'}^{S,b}$ on high and low-frequency bands are not aligned and there exists a different complex number $\kappa_{i',S}^b$. To this end, a CCC method~\cite{wei2023symbol} is performed to eliminate the complex number $\kappa_{i',S}^b$. The procedure for applying the CCC method to $\mathbf{D}_{i'}^{S,b}$ is presented in the following steps.

\textit{Step 1:} $\mathbf{D}_{i'}^{S,b}$ is divided into $N_b$ row vectors, denoted by
\begin{equation} \label{eq27}  
\mathbf{D}_{i'}^{S,b}=\left[\left(\mathbf{d}_{i',0}^{S,b}\right)^{\text{T}},\left(\mathbf{d}_{i',1}^{S,b}\right)^\text{T},\cdots,\left(\mathbf{d}_{i',N_b-1}^{S,b}\right)^\text{T}\right]^\text{T},
\end{equation}
where $\mathbf{d}_{i',n}^{S,b} \in \mathbb{C}^{1 \times M_b}$ represents the $n$-th row vector of $\mathbf{D}_{i'}^{S,b}$. 

\textit{Step 2:} A delay feature vector $\mathbf{r}_{i',S}^b \in \mathbb{C}^{N_b \times 1}$ is obtained by accumulating the results of conjugate multiplication between the row vectors, denoted by (\ref{eq28}), where $\mathbf{z}_{i'}^{b}$ is an interference vector.
\begin{figure*}
    \begin{equation}\label{eq28}
    {\fontsize{8}{8}
    \begin{aligned}   
     \mathbf{r}_{i',S}^b&=\frac{1}{M_b}\left[\frac{1}{N_b}\sum_{a=0}^{N_b-1}\mathbf{d}_{i',a}^{S,b}\left(\mathbf{d}_{i',a}^{S,b}\right)^\text{H}, \frac{1}{N_b-1}\sum_{a=0}^{N_b-1-1}\mathbf{d}_{i',a+1}^{S,b}\left(\mathbf{d}_{i',a}^{S,b}\right)^\text{H}, \cdots, \frac{1}{N_b-n}\sum_{a=0}^{N_b-n-1}\mathbf{d}_{i',a+n}^{S,b}\left(\mathbf{d}_{i',a}^{S,b}\right)^\text{H}, \cdots, \mathbf{d}_{i',N_b-1}^{S,b}\left(\mathbf{d}_{i',0}^{S,b}\right)^\text{H}\right]^\text{T} \\
     &=\left[1, e^{-j2\pi \Delta f^b \tau_{i',0}}, \cdots, e^{-j2\pi n\Delta f^b \tau_{i',0}}, \cdots, e^{-j2\pi (N_b-1) \Delta f^b \tau_{i',0}}\right]^\text{T} +\mathbf{z}_{i'}^{b}.
     \end{aligned}}
    \end{equation}
\end{figure*}
To guarantee that the SNR of $\mathbf{r}_{i',S}^b$ is evenly distributed, weighting averaging is needed to obtain the final feature vector $\mathbf{g}_{i',S}^b=\mathbf{r}_{i',S}^b\odot\mathbf{w}_r$, where $\mathbf{w}_r=[\frac{N_b}{N_b(N_b+1)/2},\frac{N_b-1}{N_b(N_b+1)/2},\cdots,\frac{1}{N_b(N_b+1)/2}]^{\text{T}}$ is the weight vector.

\textit{Step 3:} $\mathbf{D}_{i'}^{S,b}$ is divided into $M_b$ column vectors, expressed as
\begin{equation} \label{eq29}  
\mathbf{D}_{i'}^{S,b}=\left[\mathbf{c}_{i',0}^{S,b},\mathbf{c}_{i',1}^{S,b},\cdots,\mathbf{c}_{i',M_b-1}^{S,b}\right],
\end{equation}
where $\mathbf{c}_{i',m}^{S,b} \in \mathbb{C}^{N_b \times 1}$ denotes the $m$-th column vector of $\mathbf{D}_{i'}^{S,b}$. 

\textit{Step 4:} A Doppler feature vector $\mathbf{v}_{i',S}^b \in \mathbb{C}^{1 \times M_b}$ is obtained by accumulating the result of conjugate multiplication between the column vectors, expressed as (\ref{eq30}).
\begin{figure*}
    \begin{equation}\label{eq30}
    {\fontsize{8}{8}
    \begin{aligned}
  \mathbf{v}_{i',S}^b&=\frac{1}{N_b}\left[\frac{1}{M_b}\sum_{a=0}^{M_b-1}\left(\mathbf{c}_{i',a}^{S,b}\right)^\text{H}\mathbf{c}_{i',a}^{S,b}, \frac{1}{M_b-1}\sum_{a=0}^{M_b-1-1}\left(\mathbf{c}_{i',a}^{S,b}\right)^\text{H}\mathbf{c}_{i',a+1}^{S,b}, \cdots, \frac{1}{M_b-n}\sum_{a=0}^{M_b-n-1}\left(\mathbf{c}_{i',a}^{S,b}\right)^\text{H}\mathbf{c}_{i',a+n}^{S,b}, \cdots, \left(\mathbf{c}_{i',0}^{S,b}\right)^\text{H}\mathbf{c}_{i',N_b-1}^{S,b}\right] \\ 
     &= \left[1, e^{j2\pi \left(\frac{2 f_{\mathrm{C}}^bv_{i',0}}{c}\right) T^b}, \cdots, e^{j2\pi m \left(\frac{2 f_{\mathrm{C}}^bv_{i',0}}{c}\right) T^b}, \cdots, e^{j2\pi (M_b-1) \left(\frac{2 f_{\mathrm{C}}^bv_{i',0}}{c}\right) T^b}\right]+ \mathbf{\hat{z}}_{i'}^{b}.
     \end{aligned}}
    \end{equation}
     {\noindent} \rule[-10pt]{18cm}{0.1em}
\end{figure*}
Similar to delay feature vector, the final Doppler vector is $\mathbf{e}_{i',S}^b=\mathbf{v}_{i',S}^b\odot\mathbf{w}_v$, where $\mathbf{w}_v=[\frac{M_b}{M_b(M_b+1)/2},\frac{M_b-1}{M_b(M_b+1)/2},\cdots,\frac{1}{M_b(M_b+1)/2}]$ is the weight vector. 
After the above steps, the feature vector sets of the $I$ targets are obtained, namely $\{\mathbf{g}_{i',S}^1\}_{i'=1,2,\cdots,I}$, $\{\mathbf{e}_{i',S}^1\}_{i'=1,2,\cdots,I}$, $\{\mathbf{g}_{i',S}^2\}_{i'=1,2,\cdots,I}$, and $\{\mathbf{e}_{i',S}^2\}_{i'=1,2,\cdots,I}$.

\subsection{Sensing Information Fusion Stage}\label{sec3-B}
In the sensing information fusion stage, we fuse the feature vectors obtained in the signal preprocessing stage by the symbol-level fusion and estimate the parameters of $I$ targets, which include the fusion algorithm of delay feature vectors and the fusion algorithm of Doppler feature vectors.
Without loss of generality, the following remarks need to be specified.
\begin{itemize}
    \item \label{item:1} The subcarrier spacing is greater than ten times the Doppler frequency shift to ensure the orthogonality of OFDM signal~\cite{Braun2009paramether}. The Doppler frequency shift is related to the carrier frequency, so that the subcarrier spacing is proportional to the carrier frequency and we assume that $\frac{\Delta f^2}{\Delta f^1}=\left \lfloor \frac{f_{\mathrm{C}}^2}{f_{\mathrm{C}}^1} \right \rfloor = Q \ (Q>1) $, where $\left \lfloor \cdot \right \rfloor $ is flood function.
    \item \label{item:2} Upon observing Fig. \ref{fig2}, the two CCs use discrete Add CP modules, allowing for the adjustment of CP lengths. It is imperative to ensure that the maximum delay does not surpass the length of the shortest CP~\cite{Braun2009paramether,wei2023carrier}.
\end{itemize}

\subsubsection{Fusion algorithm of delay feature vectors}
Taking the $\mathbf{g}_{i',S}^1$ and $\mathbf{g}_{i',S}^2$ of the $i'$-th target as an example. Without considering the interference, observing (\ref{eq28}), the difference between $\mathbf{g}_{i',S}^1$ and $\mathbf{g}_{i',S}^2$ is the subcarrier spacing, which makes the symbol-level fusion difficult.
According to $\frac{\Delta f^2}{\Delta f^1}=Q $,  $\mathbf{g}_{i',S}^1$ and $\mathbf{g}_{i',S}^2$ are rewritten as
\begin{align}
\mathbf{g}_{i',S}^1&=\left[e^{-j2\pi n \Delta f^1 \tau_{i',0}}\right]^\text{T} |_{n=0,1,\cdots,N_1-1},\label{eq31} \\
\mathbf{g}_{i',S}^2&=\left[e^{-j2\pi nQ \Delta f^1 \tau_{i',0}}\right]^\text{T} |_{n=0,1,\cdots,N_2-1}.\label{eq32}
\end{align}
(\ref{eq32}) is further rewritten as 
\begin{equation}\label{eq33}
\hat{\mathbf{g}}_{i',S}^2=\left[e^{-j2\pi n'' \Delta f^1 \tau_{i',0}}\right]^\text{T}\in \mathbb{C}^{N_2 \times 1} |_{n''=0,Q,\cdots,Q(N_2-1)}.
\end{equation}
Then, the subcarrier spacings in (\ref{eq31}) and (\ref{eq33}) are identical. According to the MRC principle, the delay feature vectors are multiplied by the weight to carry out the subsequent symbol-level fusion, and $\tilde{\mathbf{r}}_{i',S}^1=\frac{\sigma_{S,2}^2}{\sigma_{S,1}^2+\sigma_{S,2}^2}\mathbf{g}_{i',S}^1$ and $\tilde{\mathbf{r}}_{i',S}^2=\frac{\sigma_{S,1}^2}{\sigma_{S,1}^2+\sigma_{S,2}^2}\hat{\mathbf{g}}_{i',S}^2$ are obtained. 
The fusion process of delay feature vectors is discussed below in two cases.
 
\textbf{\textit{Case 1 $(N_1-1 < Q(N_2-1))$}:} Initialize a zero-element vector $\mathbf{p}_{i',S}^1 \in \mathbb{C}^{QN_2 \times 1}$ and an index value $\xi_1 \in \{0,1,\cdots,QN_2-1\}$. Then, we traverse the element of $\mathbf{p}_{i',S}^1$ with $\xi_1$. When $\xi_1$ equals $n$, the element $\tilde{\mathbf{r}}_{i',S}^1(n)$ is assigned to $\mathbf{p}_{i',S}^1(\xi_1)$. When $\xi_1$ equals $n''$, the element $\tilde{\mathbf{r}}_{i',S}^2(n'')$ is assigned to $\mathbf{p}_{i',S}^1(\xi_1)$. Given that the zero elements in $\mathbf{p}_{i',S}^1$ deteriorate the Fourier sidelobes, which can be mitigated by recovering the missing sensing data with the CS-based sensing method proposed in~\cite{liu2023isac}. The final processed vector is denoted by $\hat{\mathbf{p}}_{i',S}^1$.

\textbf{\textit{Case 2 $\left(N_1-1\ge Q(N_2-1)\right)$}:} Initialize a zero-element vector $\mathbf{p}_{i',S}^2 \in \mathbb{C}^{N_1 \times 1}$ and an index value $\xi_2 \in \{0,1,\cdots,N_1-1\}$. Then, we traverse the element of $\mathbf{p}_{i',S}^2$ with $\xi_2$. When $\xi_2$ equals $n$, the element $\tilde{\mathbf{r}}_{i',S}^1(n)$ is assigned to $\mathbf{p}_{i',S}^2(\xi_2)$. When $\xi_2$ equals $n''$, the element $\tilde{\mathbf{r}}_{i',S}^2(n'')$ is assigned to $\mathbf{p}_{i',S}^2(\xi_2)$.

Taking the $\hat{\mathbf{p}}_{i',S}^1$ in \textit{\textbf{Case 1}} as an example, an improved IDFT method is used to achieve a high-accuracy sensing. The specific procedures are as follows, while the \textit{\textbf{Case 2}} has the similar procedures.

\textit{Step 1:} Obtain a searching interval $\left[R_\text{min}, R_\text{max}\right]$ based on the sensing demands in the application scenario of ISAC~\cite{wei2023multiple}.

\textit{Step 2:} Generate a distance searching vector $\mathbf{a}$ by gridding the searching interval with a size of grid being $\Delta R= \frac{R_\text{max}-R_\text{min}}{J}$, where $J$ is the number of girds and $\mathbf{a}\in \mathbb{C}^{J \times 1}$ is denoted by
    \begin{equation} \label{eq34}
        \mathbf{a}=\left[ R_1, R_2, \cdots, R_{\eta}, \cdots, R_J\right]^{\text{T}},
    \end{equation}
    with $\eta \in \{1,2,\cdots,J\}$ denoting the index of grids.

\textit{Step 3:} According to the size of $\hat{\mathbf{p}}_{i',S}^1$ and the expression form of delay feature vectors in (\ref{eq31}), the $\mathbf{a}$ is transformed to a searching matrix $\mathbf{A}_1 \in \mathbb{C}^{J \times QN_2}$, which is expressed as 
{\fontsize{8}{8}\begin{equation}\label{eq35}    
    \mathbf{A}_1=\left[
    \begin{array}{cccc}
     1 & e^{j2\pi  \Delta f^1 \frac{2R_1}{c}} & \cdots & e^{j2\pi (QN_2-1) \Delta f^1 \frac{2R_1}{c}} \\
     1 & e^{j2\pi  \Delta f^1 \frac{2R_2}{c}} & \cdots &  e^{j2\pi (QN_2-1) \Delta f^1 \frac{2R_2}{c}}  \\
     \vdots & \vdots &  \ddots & \vdots \\
     1 & e^{j2\pi  \Delta f^1 \frac{2R_J}{c}} & \cdots &  e^{j2\pi (QN_2-1) \Delta f^1 \frac{2R_J}{c}} \\
    \end{array} \right].
    \end{equation} }

\textit{Step 4:} Obtain and search the delay profile $\tilde{\mathbf{p}}_{i',1}=\mathbf{A}_1\hat{\mathbf{p}}_{i',S}^1$. The peak index of $\tilde{\mathbf{p}}_{i',1}$ is denoted by $\hat{\eta}_{i'}^1$.

\textit{Step 5:} The range estimation of the $i'$-th target is $\hat{r}_{i',0}^1=\mathbf{a}(\hat{\eta}_{i'}^1)$.

The fusion algorithm of delay feature vectors for the $i$-th target is shown in \hyperref[tab2]{\textbf{Algorithm 1}}, which is explained intuitively in Fig. \ref{fig3}. Finally, the estimated ranges of $I$ target are obtained.
\begin{figure}
    \centering
    \includegraphics[width=0.49\textwidth]{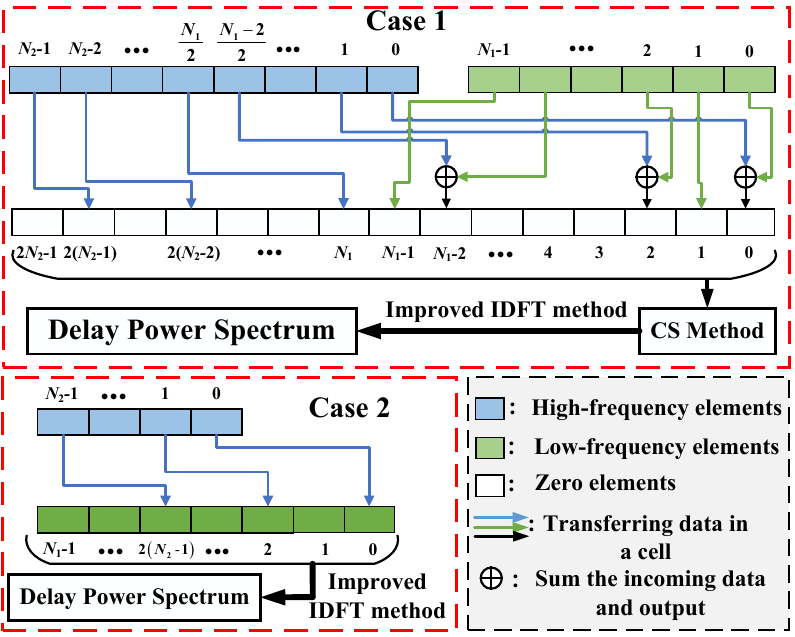}
    \caption{Fusion algorithm of
delay feature vectors when Q = 2.
}
    \label{fig3}
\end{figure}

\begin{table}[ht]
\centering
\label{tab2}
\resizebox{0.49\textwidth}{!}{
\setlength{\arrayrulewidth}{1.5pt}
\begin{tabular}{rllll}
\hline
\multicolumn{5}{l}{\textbf{Algorithm 1:} Fusion algorithm of Delay feature vectors}   \\ \hline
\multirow{-3}{*}{\textbf{Input:} }               & \multicolumn{4}{l}{\begin{tabular}[c]{@{}l@{}} Delay feature vectors of the $i'$-th target $\tilde{\mathbf{r}}_{i',S}^1$ and $\tilde{\mathbf{r}}_{i',S}^2$;\\ The number of subcarriers $N_1$ and $N_2$;\\ The subcarriers spacing $\Delta f^1$ and $\Delta f^2$. \end{tabular}} \\
\textbf{Output:}               & \multicolumn{4}{l}{The range estimations of the $i'$-th target $\hat{r}_{i',0}^1$ and $\hat{r}_{i',0}^2$.} \\ 
1:      & \multicolumn{4}{l}{$\textbf{if}$ $N_1-1 \ge Q(N_2-1)$ $\textbf{do}$}   \\
2:      & \multicolumn{4}{l}{$\hspace{1em}$ Assume a vector $\mathbf{p}_{i',S}^2 \in \mathbb{C}^{N_1\times 1}=\tilde{\mathbf{r}}_{i',S}^1$,}\\ &\multicolumn{4}{l}{$\hspace{1em}$ and an index value $\mu_2 \in \mathcal{N}_2=\{0,1,\cdots,N_2-1\}$;}   \\
3:      & \multicolumn{4}{l}{$\hspace{1em}$ $\textbf{for}$ $\mu_2$ in $\mathcal{N}_2$ $\textbf{do}$}   \\
4:       & \multicolumn{4}{l}{$\hspace{2em}$ The element value of $\tilde{\mathbf{r}}_{i',S}^2(\mu_2)$ is added to $\mathbf{p}_{i',S}^2(Q\mu_2)$;}   \\
5:       & \multicolumn{4}{l}{$\hspace{2em}$ $\mathbf{p}_{i',S}^2(Q\mu_2)$ =$\mathbf{p}_{i',S}^2(Q\mu_2)/2$;}  \\
6:       & \multicolumn{4}{l}{$\hspace{1em}$ $\textbf{end}$ $\textbf{for}$}  \\
7:      & \multicolumn{4}{l}{$\textbf{else if}$ $N_1-1 < Q(N_2-1)$ $\textbf{do}$}   \\
\multirow{-1}{*}{8:}        & \multicolumn{4}{l}{$\hspace{1em}$ Initialize a zero-element vector $\mathbf{p}_{i',S}^1 \in \mathbb{C}^{QN_2 \times 1}$,}\\ &\multicolumn{4}{l}{$\hspace{1em}$ and an index value $\mu_1 \in \mathcal{N}_1=\{0,1,\cdots,N_1-1\}$;} \\
9:      & \multicolumn{4}{l}{$\hspace{1em}$ $\textbf{for}$ $\mu_1$ in $\mathcal{N}_1$ $\textbf{do}$}   \\
10:        & \multicolumn{4}{l}{$\hspace{2em}$ The element value of $\tilde{\mathbf{r}}_{i',S}^1(\mu_1)$ is assigned to $\mathbf{p}_{i',S}^1(\mu_1)$;}  \\
11:        & \multicolumn{4}{l}{$\hspace{1em}$ $\textbf{end}$ $\textbf{for}$}  \\
12:        & \multicolumn{4}{l}{$\hspace{1em}$ Initialize an index value $\mu_3 \in \mathcal{N}_3=\{0,1,\cdots,N_2-1\}$;}  \\
13:       & \multicolumn{4}{l}{$\hspace{1em}$ $\textbf{for}$ $\mu_3$ in $\mathcal{N}_3$ $\textbf{do}$}   \\
14:       & \multicolumn{4}{l}{$\hspace{2em}$ The element value of $\tilde{\mathbf{r}}_{i',S}^2(\mu_3)$ is added to $\mathbf{p}_{i',S}^1(Q\mu_3)$;}   \\
15:       & \multicolumn{4}{l}{$\hspace{2em}$ $\textbf{if}$ $Q\mu_3 \le N_1-1 $ $\textbf{do}$}  \\
16:       & \multicolumn{4}{l}{$\hspace{3em}$ $\mathbf{p}_{i',S}^1(Q\mu_3)$ =$\mathbf{p}_{i',S}^1(Q\mu_3)/2$;}  \\
17:       & \multicolumn{4}{l}{$\hspace{2em}$ $\textbf{end}$ $\textbf{if}$}  \\
18:      & \multicolumn{4}{l}{$\hspace{1em}$ $\textbf{end}$ $\textbf{for}$}  \\
19:      & \multicolumn{4}{l}{$\hspace{1em}$ Recover the missing elements of $\mathbf{p}_{i',S}^1$ by }\\ & \multicolumn{4}{l}{$\hspace{1em}$ CS-based method proposed in~\cite{liu2023isac} to obtain $\hat{\mathbf{p}}_{i',S}^1$;}  \\
20:       & \multicolumn{4}{l}{$\textbf{end}$ $\textbf{if}$}  \\
21:      & \multicolumn{4}{l}{The results of $\hat{\mathbf{p}}_{i',S}^1$ and $\mathbf{p}_{i',S}^2$ undergoing improved IDFT are}\\& \multicolumn{4}{l}{ searched for peaks to obtain the peak indices $\hat{\eta}_{i'}^1$ and $\hat{\eta}_{i'}^2$;}  \\
22:      & \multicolumn{4}{l}{Substitute $\hat{\eta}_{i'}^1$ and $\hat{\eta}_{i'}^2$ into (\ref{eq34}) to obtain the range}\\  & \multicolumn{4}{l}{  estimations of the $i'$-th target $\hat{r}_{i',0}^1$ and $\hat{r}_{i',0}^2$.}  \\
\hline
\end{tabular}}
\end{table}

\subsubsection{Fusion algorithm of Doppler feature vectors}
Taking the $\mathbf{e}_{i',S}^1$ and $\mathbf{e}_{i',S}^2$ as an example. Without considering the interference, observing (\ref{eq30}), the different parameters between $\mathbf{e}_{i',S}^1$ and $\mathbf{e}_{i',S}^2$ are carrier frequency and total symbol duration. (\ref{eq30}) is rewritten as
\begin{equation} \label{eq36}
\mathbf{e}_{i',S}^b=\left[e^{j2\pi m \frac{2 v_{i',0}}{c}\left(f_{\mathrm{C}}^bT^b\right) }\right]|_{m=0,1,\cdots,M_b-1},
\end{equation}
where $f_{\mathrm{C}}^bT^b=f_{\mathrm{C}}^b\left(\frac{1}{\Delta f^b}+T_\text{s}^b\right)$. If $f_{\mathrm{C}}^1T^1=f_{\mathrm{C}}^2T^2$, $\mathbf{e}_{i',S}^1$ and $\mathbf{e}_{i',S}^2$ can be coherently accumulated to obtain a SNR gain.

Fortunately, as shown in Section \ref{sec3-B}, the unique framework of the proposed CA-enabled MIMO-OFDM ISAC system supports adjustable durations of CPs in the two CCs. Meanwhile, \hyperref[theorem 1]{\textbf{Theorem 1}} proves that $f_{\mathrm{C}}^1T^1=f_{\mathrm{C}}^2T^2$.
\begin{theorem}\label{theorem 1}
    According to $\left \lfloor \frac{f_{\mathrm{C}}^2}{f_{\mathrm{C}}^1} \right \rfloor = Q$, we assume that $\frac{f_{\mathrm{C}}^2}{f_{\mathrm{C}}^1}=Q+\rho (\rho <1)$. $f_{\mathrm{C}}^1T^1=f_{\mathrm{C}}^2T^2$ can be satisfied when 
    \begin{equation} \label{eq37}
        N_\text{s}^1=\left(1+\frac{\rho}{Q}\right)\frac{N_1}{N_2}N_\text{s}^2+\frac{N_1\rho}{Q}.
    \end{equation}
\end{theorem}
\begin{proof}
    Based on $\frac{f_{\mathrm{C}}^2}{f_{\mathrm{C}}^1}=\frac{\Delta f^2}{\Delta f^1}+\rho$, the relationship between $T_{\text{s}}^1$ and $T_{\text{s}}^2$ is expressed as
    \begin{equation} \label{eq38}{\fontsize{8}{8}
    \begin{aligned}
        f_{\mathrm{C}}^1\left(\frac{1}{\Delta f^1}+T_\text{s}^1\right) &=f_{\mathrm{C}}^2\left(\frac{1}{\Delta f^2}+T_\text{s}^2\right),\\
        \Rightarrow \left(\frac{f_{\mathrm{C}}^1}{\Delta f^1}+ f_{\mathrm{C}}^1T_\text{s}^1\right) &=f_{\mathrm{C}}^1\left(\frac{\Delta f^2}{\Delta f^1}+\rho\right)\left(\frac{1}{\Delta f^2}+T_\text{s}^2\right),\\
        \Rightarrow
        T_\text{s}^1&=(Q+\rho)T_\text{s}^2+\frac{\rho}{\Delta f^2}.
    \end{aligned} }
    \end{equation}
    Then, according to $T_\text{s}^b=\frac{N_\text{s}^bT_{\text{ofdm}}^b}{N_b}$, the relationship between $N_\text{s}^1$ and $N_\text{s}^2$ can be expressed as 
    \begin{equation} \label{eq39}
        \begin{aligned}
            \frac{N_\text{s}^1T_{\text{ofdm}}^1}{N_1}&=(Q+\rho)\frac{N_\text{s}^2T_{\text{ofdm}}^2}{N_2}+\frac{\rho}{\Delta f^2},\\
            \Rightarrow
            N_\text{s}^1&=(\frac{Q+\rho}{Q})\frac{N_1}{N_2}N_\text{s}^2+\frac{N_1\rho}{Q}.
        \end{aligned}
    \end{equation}
\end{proof}

Similar to the fusion of delay feature vectors, the $\tilde{\mathbf{v}}_{i',S}^1=\frac{\sigma_{S,2}^2}{\sigma_{S,1}^2+\sigma_{S,2}^2}\mathbf{e}_{i',S}^1$ and $\tilde{\mathbf{v}}_{i',S}^2=\frac{\sigma_{S,1}^2}{\sigma_{S,1}^2+\sigma_{S,2}^2}\mathbf{e}_{i',S}^2$ are obtained. The fusion process of Doppler feature vectors is revealed in the following three cases.

\textbf{\textit{Case 1 $(M_1 > M_2)$}:} Initialize a zero-element vector $\mathbf{f}_{i',S}^1 \in \mathbb{C}^{1 \times M_1}$ and complement $\tilde{\mathbf{v}}_{i',S}^2$ with zero to get the matrix $\hat{\mathbf{v}}_{i',S}^2 \in \mathbb{C}^{1 \times M_1}=[\tilde{\mathbf{v}}_{i',S}^2,0_{M_2},\cdots,0_{M_1-1}]$. Then, the processed vector $\mathbf{f}_{i',S}^1=(\tilde{\mathbf{v}}_{i',S}^1+\hat{\mathbf{v}}_{i',S}^2)/2$ is obtained.

\textbf{\textit{Case 2 $(M_1 < M_2)$}:} Initialize a zero-element vector $\mathbf{f}_{i',S}^2 \in \mathbb{C}^{1 \times M_2}$ and complement $\tilde{\mathbf{v}}_{i',S}^1$ with zero to get the matrix $\hat{\mathbf{v}}_{i',S}^1 \in \mathbb{C}^{1 \times M_2}=[\tilde{\mathbf{v}}_{i',S}^1,0_{M_1},\cdots,0_{M_2-1}]$. Then, the processed vector $\mathbf{f}_{i',S}^2=(\hat{\mathbf{v}}_{i',S}^1+\tilde{\mathbf{v}}_{i',S}^2)/2$ is obtained.

\textbf{\textit{Case 3 $(M_1= M_2=M)$}:} Initialize a zero-element vector $\mathbf{f}_{i',S}^3 \in \mathbb{C}^{1 \times M}$ and obtain a vector $\mathbf{f}_{i',S}^3=(\tilde{\mathbf{v}}_{i',S}^1+\tilde{\mathbf{v}}_{i',S}^2)/2$.

Taking the $\mathbf{E}_{i',S}^1$ in \textbf{\textit{Case 1}} as an example, an improved discrete Fourier transform (DFT) method is applied to achieve high-accuracy sensing. The specific procedures are as follows, while the \textbf{\textit{Case 2}} and \textbf{\textit{Case 3}} have the similar procedures. 

\textit{Step 1:} Obtain a searching interval $\left[V_\text{min},V_\text{max}\right]$ based on the sensing demands in the application scenario of ISAC~\cite{wei2023multiple}.

\textit{Step 2:} Generate a velocity searching vector $\mathbf{b}$ by gridding the searching interval with a size of grid being $\Delta V=\frac{V_\text{max}-V_\text{min}}{G}$, where $G$ is the number of grids and $\mathbf{b} \in \mathbb{C}^{1\times G}$ is expressed as
    \begin{equation} \label{eq40}
      \mathbf{b}=\left[V_1,V_2,
     \cdots,V_g,\cdots,V_G\right],
    \end{equation}
    with $g \in \{1,2,\cdots,G\}$ representing the index of grids.

\textit{Step 3:} According to the size of $\mathbf{f}_{i',S}^1$ and the expression form of Doppler feature vectors in (\ref{eq30}), $\mathbf{b}$ is transformed to a searching matrix $\mathbf{B}_1 \in \mathbb{C}^{M_1 \times G}$, which is expressed as 
    \begin{equation}\label{eq41} {
    \begin{aligned}
        \mathbf{B}_1=\left[
    \begin{array}{cccc}
     1 & e^{-j2\pi  f_\mathrm{C}^1 T^1 \frac{2V_1}{c}} & \cdots & e^{-j2\pi (M_1-1) f_\mathrm{C}^1 T^1 \frac{2V_1}{c}} \\
     1 & e^{-j2\pi f_\mathrm{C}^1 T^1 \frac{2V_2}{c}} & \cdots & e^{-j2\pi (M_1-1) f_\mathrm{C}^1 T^1 \frac{2V_2}{c}} \\
     \vdots & \vdots &  \ddots & \vdots \\
     1 & e^{-j2\pi f_\mathrm{C}^1 T^1 \frac{2V_G}{c}} & \cdots & e^{-j2\pi (M_1-1) f_\mathrm{C}^1 T^1 \frac{2V_G}{c}} \\
    \end{array} \right]^{\text{T}}.
    \end{aligned} }
    \end{equation}

\textit{Step 4:} Obtain and search the Doppler profile $\tilde{\mathbf{f}}_{i',1}=\mathbf{f}_{i',S}^1\mathbf{B}_1$. The peak index of $\tilde{\mathbf{f}}_{i',1}$ is denoted by $\hat{g}_{i'}^1$.

\textit{Step 5:} The velocity estimation of the $i'$-th target is $\hat{v}_{i',0}^1=\mathbf{b}(\hat{g}_{i'}^1)$.

The fusion algorithm of Doppler feature vectors for the $i'$-th target is shown in \hyperref[tab3]{\textbf{Algorithm 2}}, which is explained intuitively in Fig. \ref{fig4}. Finally, the estimated velocities of $I$ targets are obtained.
\begin{table}[!ht]
\centering
\label{tab3}
\resizebox{0.49\textwidth}{!}{
\setlength{\arrayrulewidth}{1.5pt}
\begin{tabular}{rllll}
\hline
\multicolumn{5}{l}{\textbf{Algorithm 2 :} Fusion algorithm of Doppler feature vectors}   \\ \hline
\multirow{-4}{*}{\textbf{Input:} }               & \multicolumn{4}{l}{\begin{tabular}[c]{@{}l@{}}Doppler feature vectors of the $i'$-th target $\tilde{\mathbf{v}}_{i',S}^1$ and $\tilde{\mathbf{v}}_{i',S}^2$;\\ The number of OFDM symbols $M_1$ and $M_2$;\\ The carrier frequencies $f_\mathrm{C}^1$ and $f_\mathrm{C}^2$;\\ The total OFDM duration $T^1$ and $T^2$. \end{tabular}} \\
\textbf{Output:}               & \multicolumn{4}{l}{The velocity estimations of the $i'$-th target $\hat{v}_{i',0}^1$, $\hat{v}_{i',0}^2$, and $\hat{v}_{i',0}^3$.} \\ 
1:      & \multicolumn{4}{l}{$\textbf{if}$ $M_1 > M_2$ $\textbf{do}$}   \\
2:        & \multicolumn{4}{l}{$\hspace{1em}$ Assume a vector $\mathbf{f}_{i',S}^1 \in \mathbb{C}^{1 \times M_1}= \tilde{\mathbf{v}}_{i',S}^1$,} \\  & \multicolumn{4}{l}{$\hspace{1em}$ and an index value $\omega_1 \in \mathcal{M}_1=\{0,1,\cdots, M_2-1\}$;} \\
3:      & \multicolumn{4}{l}{$\hspace{1em}$ $\textbf{for}$ $\omega_1$ in $\mathcal{M}_1$ $\textbf{do}$}   \\
4:      & \multicolumn{4}{l}{$\hspace{2em}$ The element value of $\tilde{\mathbf{v}}_{i',S}^2(\omega_1)$ is added to $\mathbf{f}_{i',S}^1(\omega_1)$;}   \\
5:        & \multicolumn{4}{l}{$\hspace{2em}$ $\mathbf{f}_{i',S}^1(\omega_1)=\mathbf{f}_{i',S}^1(\omega_1)/2$;}  \\
6:        & \multicolumn{4}{l}{$\hspace{1em}$ $\textbf{end}$ $\textbf{for}$}  \\
7:       & \multicolumn{4}{l}{$\textbf{else if}$ $M_1 < M_2$ $\textbf{do}$}   \\
8:       & \multicolumn{4}{l}{$\hspace{1em}$ Assume a vector $\mathbf{f}_{i',S}^2 \in \mathbb{C}^{1 \times M_2}= \tilde{\mathbf{v}}_{i',S}^2$,} \\  & \multicolumn{4}{l}{$\hspace{1em}$ and an index value $\omega_2 \in \mathcal{M}_2=\{0,1,\cdots, M_1-1\}$;}   \\
9:      & \multicolumn{4}{l}{$\hspace{1em}$ $\textbf{for}$ $\omega_2$ in $\mathcal{M}_2$ $\textbf{do}$}   \\
10:      & \multicolumn{4}{l}{$\hspace{2em}$ The element value of $\tilde{\mathbf{v}}_{i',S}^1(\omega_2)$ is added to $\mathbf{f}_{i',S}^2(\omega_2)$;}   \\
11:        & \multicolumn{4}{l}{$\hspace{2em}$ $\mathbf{f}_{i',S}^2(\omega_2)=\mathbf{f}_{i',S}^2(\omega_2)/2$;}  \\
12:        & \multicolumn{4}{l}{$\hspace{1em}$ $\textbf{end}$ $\textbf{for}$}  \\
13:       & \multicolumn{4}{l}{$\textbf{else if}$ $M_1 = M_2 =M$ $\textbf{do}$}   \\
14:       & \multicolumn{4}{l}{$\hspace{1em}$ Initialize a zero-elements vector $\mathbf{f}_{i',S}^3 \in \mathbb{C}^{1 \times M}$,} \\  & \multicolumn{4}{l}{$\hspace{1em}$ and an index value $\omega_3 \in \mathcal{M}_3=\{0,1,\cdots, M-1\}$;}   \\
15:      & \multicolumn{4}{l}{$\hspace{1em}$ $\textbf{for}$ $\omega_3$ in $\mathcal{M}_3$ $\textbf{do}$}   \\
16:      & \multicolumn{4}{l}{$\hspace{2em}$ $\mathbf{f}_{i',S}^3(\omega_3)=\left[\tilde{\mathbf{v}}_{i',S}^1(\omega_3)+\tilde{\mathbf{v}}_{i',S}^2(\omega_3)\right]/2$;}   \\
17:        & \multicolumn{4}{l}{$\hspace{1em}$ $\textbf{end}$ $\textbf{for}$}  \\

18:      & \multicolumn{4}{l}{$\textbf{end}$ $\textbf{if}$}  \\
19:      & \multicolumn{4}{l}{The results of $\mathbf{f}_{i',S}^1$, $\mathbf{f}_{i',S}^2$, and $\mathbf{f}_{i',S}^3$ undergoing improved DFT are}\\& \multicolumn{4}{l}{ searched for peaks to obtain the peak indices $\hat{g}_1$, $\hat{g}_2$, and $\hat{g}_3$;}  \\
20:      & \multicolumn{4}{l}{Substitute $\hat{g}_{i'}^1$, $\hat{g}_{i'}^2$, and $\hat{g}_{i'}^3$ into (\ref{eq40}) to obtain the velocity}\\  & \multicolumn{4}{l}{ estimations of 
 the $i'$-th target $\hat{v}_{i',0}^1$, $\hat{v}_{i',0}^2$, and $\hat{v}_{i',0}^3$.}  \\
\hline
\end{tabular}}
\end{table}
\begin{figure}[!ht]
    \centering
    \includegraphics[width=0.49\textwidth]{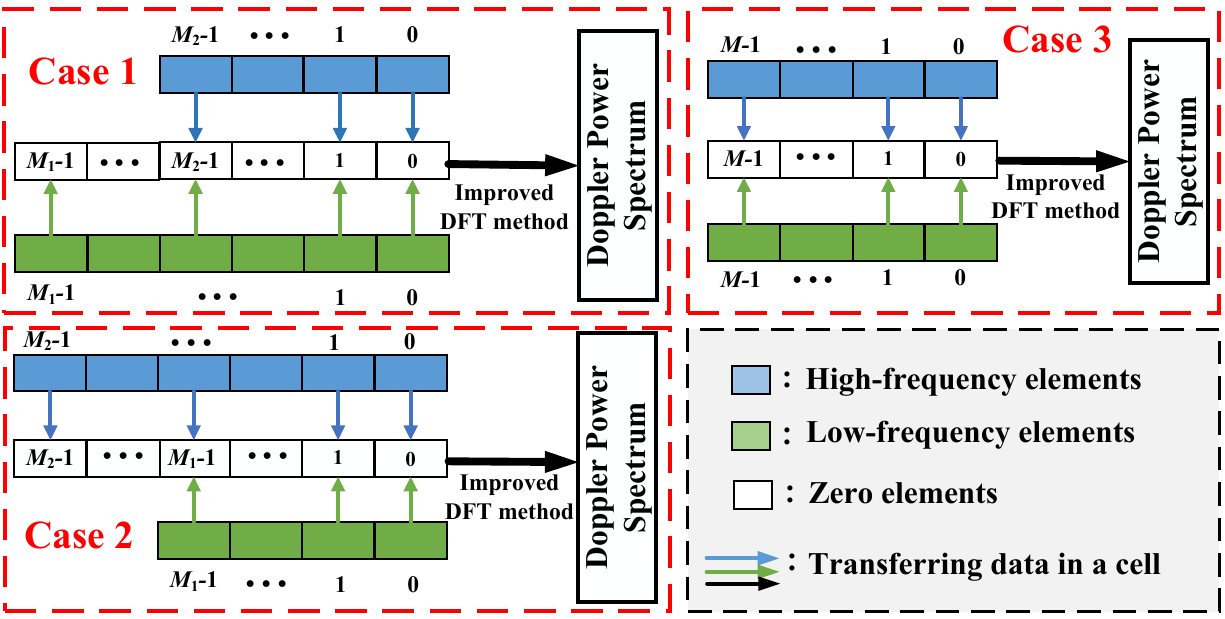}
    \caption{Fusion algorithm of Doppler feature vectors.}
    \label{fig4}
\end{figure}

\section{Performance Analysis} \label{se4}
In this section, communication MI is used to characterize the channel capacity of the proposed ISAC system, and CRLB is used as a performance metric of sensing for the proposed ISAC system~\cite{ouyang2023integrated}.

\subsection{Communication Performance Analysis} \label{se4-A}
In this section, the closed-form communication MI of the proposed CA-enabled MIMO-OFDM ISAC system is derived with the assumption of known communication CSI at both the BS and UE.
\subsubsection{Communication MI of MIMO-OFDM ISAC system}
The communication MI for MIMO-OFDM ISAC system is provided in \hyperref[lemma 1]{\textbf{Lemma 1}}.
\begin{lemma} \label{lemma 1}
  According to the signal model in (\ref{eq42}), the communication MI of the MIMO-OFDM ISAC system is obtained in (\ref{eq43}), with $N_x$, $N_c$, and $N_r$ denoting the number of OFDM symbols, subcarriers, and receive antennas, respectively \cite{Wei2023MI}.
    \begin{equation}\label{eq42}
        \mathbf{Y}_\text{com}=\mathbf{X}\mathbf{H}+\mathbf{W}_\text{com},
    \end{equation}
    \begin{equation}\label{eq43}
        {\fontsize{10}{10}\begin{aligned}
       &I(\mathbf{X};\mathbf{Y}_\text{com}{ | \mathbf{H})} \\
       &=N_x\mathrm{log}_2\left[\left(\sigma_\text{Z}^2\right)^{-N_rN_c}\det(\mathbf{H}^\text{H}\frac{E\{\mathbf{X}^\text{H}\mathbf{X}\}}{N_x}\mathbf{H}+\sigma_\text{Z}^2\mathbf{I}_{N_rN_c})\right],
       \end{aligned}}
    \end{equation}
    where $\mathbf{W}_\text{com} \sim \mathcal{CN}\left(0,\sigma_\text{Z}^2\right)$.  
\end{lemma}
\subsubsection{Communication MI of the proposed CA-enabled MIMO-OFDM ISAC system}
According to \cite{Wei2023MI} and (\ref{eq8}), the communication MI of the proposed CA-enabled MIMO-OFDM ISAC system is defined as
\begin{equation} \label{eq44}
I\left(\mathbf{X}_\mathrm{C}^1,\mathbf{X}_\mathrm{C}^2;\mathbf{Y}_u^\mathrm{C}\mid\mathbf{H}_\mathrm{C}^1,\mathbf{H}_\mathrm{C}^2\right).
\end{equation}
Since the frequency bands between the two CCs are far separated and transmit data independently, $\mathbf{X}_\mathrm{C}^1$ and $\mathbf{X}_\mathrm{C}^2$ are independent, as well as $\mathbf{H}_\mathrm{C}^1$ and $\mathbf{H}_\mathrm{C}^2$. Therefore, (\ref{eq44}) is rewritten as
\begin{equation} \label{eq45}
{\fontsize{8.5}{10}
\begin{aligned}
&I\left(\mathbf{X}_\mathrm{C}^1,\mathbf{X}_\mathrm{C}^2;\mathbf{Y}_u^\mathrm{C}\mid\mathbf{H}_\mathrm{C}^1,\mathbf{H}_\mathrm{C}^2\right)\\
&=I\left(\mathbf{X}_\mathrm{C}^{1};\mathbf{Y}_{u}^\mathrm{C}\mid\mathbf{H}_\mathrm{C}^{1},\mathbf{H}_\mathrm{C}^{2}\right)+I\left(\mathbf{X}_\mathrm{C}^{2};\mathbf{Y}_{u}^\mathrm{C}\mid\mathbf{X}_\mathrm{C}^{1},\mathbf{H}_\mathrm{C}^{1},\mathbf{H}_\mathrm{C}^{2}\right) \\
&=h\left(\mathbf{Y}_u^\mathrm{C}\mid\mathbf{H}_\mathrm{C}^1,\mathbf{H}_\mathrm{C}^2\right)-h\left(\mathbf{Y}_u^\mathrm{C}\mid\mathbf{X}_\mathrm{C}^1,\mathbf{H}_\mathrm{C}^1,\mathbf{H}_\mathrm{C}^2\right) \\ & \quad+h\left(\mathbf{Y}_u^\mathrm{C}\mid\mathbf{X}_\mathrm{C}^1,\mathbf{H}_\mathrm{C}^1,\mathbf{H}_\mathrm{C}^2\right)-h\left(\mathbf{Y}_u^\mathrm{C}\mid\mathbf{X}_\mathrm{C}^1,\mathbf{X}_\mathrm{C}^2,\mathbf{H}_\mathrm{C}^1,\mathbf{H}_\mathrm{C}^2\right)\\
&=h{\left(\mathbf{Y}_{u}^\mathrm{C}\mid\mathbf{H}_\mathrm{C}^{1},\mathbf{H}_\mathrm{C}^{2}\right)}-h{\left(\mathbf{W}_{u}^\mathrm{C}\right)}.
\end{aligned}}
\end{equation}
\begin{theorem}\label{theorem 2}
    According to (\ref{eq8}) and (\ref{eq45}), the communication MI of the proposed CA-enabled MIMO-OFDM ISAC system is expressed as
      \begin{equation}\label{eq46}
      {\fontsize{6}{8} \begin{aligned}
&I\left(\mathbf{X}_\mathrm{C}^1,\mathbf{X}_\mathrm{C}^2;\mathbf{Y}_u^\mathrm{C}\mid\mathbf{H}_\mathrm{C}^1,\mathbf{H}_\mathrm{C}^2\right)= \\& M\log_2\left[\left(\sigma_\mathrm{C}^2\right)^{-N_\text{U}N}\det\left(\sum_{b=1}^\mathcal{B}\left[\left(\mathbf{H}^b\right)^\text{H}E\left[\frac{\left(\mathbf{X}_\mathrm{C}^b\right)^\text{H}\mathbf{X}_\mathrm{C}^b}{M}\right]\mathbf{H}^b\right]+\sigma_\mathrm{C}^2\mathbf{I}_{N_\text{U}N}\right)\right].
     \end{aligned}} 
     \end{equation}
\end{theorem}
\begin{proof}
    Without loss of generality, we consider the assumptions of $N_1=N_2=N$ and $M_1=M_2=M$. Since $\mathbf{X}_\mathrm{C}^{b}$ is a block diagonal array, the probability density function (PDF) of $\mathbf{Y}_u^\mathrm{C}$ under known CSI is the product of the PDFs on the $N$ subcarriers, denoted by (\ref{eq47}), where $\mathbf{y}_u^{\mathrm{C},j}$ is the $j$-th row or column vector of $\mathbf{Y}_u^\mathrm{C}$, and $\mathbf{I}_{N_{\text{U}}}$ represents a $N_{\text{U}} \times N_{\text{U}}$ unit matrix.
    \begin{figure*}
        \centering
        \begin{equation}
            \label{eq47}
           {\fontsize{7.5}{8} \begin{aligned}
&p\left(\mathbf{Y}_u^\mathrm{C}\mid\mathbf{H}_\mathrm{C}^1,\mathbf{H}_\mathrm{C}^2\right)=\prod_{j=1}^{NM}p\left(\mathbf{y}_u^{\mathrm{C},j}\mid\mathbf{H}_\mathrm{C}^1,\mathbf{H}_\mathrm{C}^2\right) \\
&=\prod_{j=1}^{NM}\frac{1}{\pi^{N_\text{U}}\text{det}\left[\sum\limits_{b=1}^{\mathcal{B}}\left[\left(\mathbf{H}_\mathrm{C}^b\right)^\text{H}E\left\{\left(\mathbf{X}_\mathrm{C}^{b, j}\right)^\text{H}\mathbf{X}_\mathrm{C}^{b, j}\right\}\mathbf{H}_\mathrm{C}^b\right]+\sigma_\mathrm{C}^2\mathbf{I}_{N_\text{U}}\right]}\exp\left\{-\left[\left(\sum\limits_{b=1}^{\mathcal{B}}\left[\left(\mathbf{H}_\mathrm{C}^{b}\right)^{\text{H}}E\left\{\left(\mathbf{X}_\mathrm{C}^{b, j}\right)^{\text{H}}\mathbf{X}_\mathrm{C}^{b, j}\right\}\mathbf{H}_\mathrm{C}^{b}\right]+\sigma_\mathrm{C}^{2}\mathbf{I}_{N_{\text{U}}}\right)^{-1}\mathbf{y}_{u}^{\mathrm{C},j}\left(\mathbf{y}_{u}^{\mathrm{C},j}\right)^{\text{H}}\right]\right\} \\
&=\prod_{j=1}^{N}\frac1{\pi^{MN_{\text{U}}}\det^{M}\left[\sum\limits_{b=1}^{\mathcal{B}}\left[\left(\tilde{\mathbf{H}}_{u,m}^{n}\right)^{\text{H}}E\left\{\left(\mathbf{X}_{n}^{b,j}\right)^{\text{H}}\mathbf{X}_{n}^{b,j}\right\}\tilde{\mathbf{H}}_{a,m}^{n}\right]+\sigma_\mathrm{C}^{2}\mathbf{I}_{N_{\text{U}}}\right]} \\ & \quad \times\exp\left\{-\text{tr}\left[\sum\limits_{j=1}^{M}\left(\sum\limits_{b=1}^{\mathcal{B}}\left[\left(\tilde{\mathbf{H}}_{u,m}^{n}\right)^{\text{H}}E\left\{\left(\mathbf{X}_{n}^{b,j}\right)^{\text{H}}\mathbf{X}_{n}^{b,j}\right\}\tilde{\mathbf{H}}_{u,m}^{n}\right]+\sigma_\mathrm{C}^{2}\mathbf{I}_{N_{\text{U}}}\right)^{-1}\mathbf{y}_{u}^{\mathrm{C},j}\left(\mathbf{y}_{u}^{\mathrm{C},j}\right)^{\text{H}}\right]\right\},
\end{aligned}}
        \end{equation}
    \end{figure*}  
According to (\ref{eq47}), the entropy matrix $h\left(\mathbf{Y}_u^\mathrm{C}\mid\mathbf{H}_\mathrm{C}^1,\mathbf{H}_\mathrm{C}^2\right)$ is obtained as in (\ref{eq48}).
\begin{figure*}
 \begin{equation} \label{eq48}
 \begin{aligned}
& h{\left(\mathbf{Y}_u^\mathrm{C}\mid\mathbf{H}_\mathrm{C}^1,\mathbf{H}_\mathrm{C}^2\right)}\\ & =\sum_{n=0}^{N-1} \bigg\{MN_\text{U}{ \log _ 2 \pi + M }N_\text{U}{ \log _ 2 e}  +M\log_{2}\left[\det\left(\sum_{b=1}^{\mathcal{B}}\left[\left(\tilde{\mathbf{H}}_{u,m}^{n}\right)^{\text{H}}E\left\{\left(\mathbf{X}_{n}^{b,j}\right)^{\text{H}}\mathbf{X}_{n}^{b,j}\right\}\mathbf{\tilde{H}}_{u,m}^{n}\right]+\sigma_\mathrm{C}^{2}\mathbf{I}_{N_{\text{U}}}\right)\right]\bigg\}.
\end{aligned}
\end{equation}   
\end{figure*}
Meanwhile, the noise entropy $h\left(\mathbf{W}_u^\mathrm{C}\right)$ can be obtained as
\begin{equation} \label{eq49} 
\begin{aligned}
 h{\left(\mathbf{W}_u^\mathrm{C}\right)}&=NMN_\mathrm{U}\log_2\pi+NMN_\mathrm{U}\log_2e\\ & \quad+NM\mathrm{log}_2{\left[\det{\left(\sigma_\mathrm{C}^2\mathbf{I}_{N_\text{U}}\right)}\right]}.   
\end{aligned}
\end{equation}

Therefore, substituting (\ref{eq48}) and (\ref{eq49}) into (\ref{eq45}), the communication MI of the proposed CA-enabled MIMO-OFDM ISAC system in (\ref{eq46}) is derived in (\ref{eq50}),
\begin{figure*}
  \begin{equation} \label{eq50}
    {   \begin{aligned}
&I\left(\mathbf{X}_\mathrm{C}^1,\mathbf{X}_\mathrm{C}^2;\mathbf{Y}_u^\mathrm{C}\mid\mathbf{H}_\mathrm{C}^1,\mathbf{H}_\mathrm{C}^2\right) \\
&=M\sum_{n=0}^{N-1}\log_{2}\left[\left(\sigma_\mathrm{C}^{2}\right)^{-N_{\text{U}}}\det\left(\sum_{b=1}^{\mathcal{B}}\left[\left(\tilde{\mathbf{H}}_{u,m}^{n}\right)^{\text{H}}E\left\{\left(\mathbf{X}_{n}^{b,j}\right)^{\mathrm{H}}\mathbf{X}_{n}^{b,j}\right\}\mathbf{\tilde{H}}_{u,m}^{n}\right]+\sigma_\mathrm{C}^{2}\mathbf{I}_{N_{\text{U}}}\right)\right] \\
&=M\log_2\left[\left(\sigma_\mathrm{C}^2\right)^{-N_\text{U}N}\prod_{n=0}^{N-1}\det\left(\sum_{b=1}^{\mathcal{B}}\left[\left(\tilde{\mathbf{H}}_{u,m}^n\right)^\text{H}E\left\{\left(\mathbf{X}_n^{b,j}\right)^\text{H}\mathbf{X}_n^{b,j}\right\}\mathbf{\tilde{H}}_{u,m}^n\right]+\sigma_\mathrm{C}^2\mathbf{I}_{N_\text{U}}\right)\right] \\
&=M\log_2\left[\left(\sigma_\mathrm{C}^2\right)^{-N_\text{U}N}\det\left(\sum_{b=1}^{\mathcal{B}}\left[\left(\mathbf{H}^b\right)^\text{H}\text{diag}\left\{E\left\{\left(\mathbf{X}_\mathrm{C}^{b,j}\right)^\text{H}\mathbf{X}_\mathrm{C}^{b,j}\right\}\right\}\mathbf{H}^b\right]+\sigma_\mathrm{C}^2\mathbf{I}_{N_\text{U}N}\right)\right] \\
&=M\log_2\left[\left(\sigma_\mathrm{C}^2\right)^{-N_\text{U}N}\det\left[\sum_{b=1}^{\mathcal{B}}\left[\left(\mathbf{H}^b\right)^\text{H}E\left\{\frac{\left(\mathbf{X}_\mathrm{C}^b\right)^\text{H}\mathbf{X}_\mathrm{C}^b}M\right\}\mathbf{H}^b\right]+\sigma_\mathrm{C}^2\mathbf{I}_{N_\text{U}N}\right)\right],
\end{aligned}}
\end{equation}
     {\noindent} \rule[-10pt]{18cm}{0.1em}
\end{figure*}
where the $\mathbf{H}^b$ is the transformation of $\mathbf{H}_\mathrm{C}^b$ in (\ref{eq8}), expressed as follows
\begin{equation}\label{eq51}
    \mathbf{H}^b=\text{diag}\left\{\mathbf{\tilde{H}}_{u,m}^0,\cdots,\mathbf{\tilde{H}}_{u,m}^{N-1}\right\}\in\mathbb{C}^{N_\text{T}N\times N_\text{U}N}.
\end{equation}
\end{proof}

Upon observing (\ref{eq46}), we find that the maximum communication channel capacity of the proposed ISAC system is affected by space-time-frequency resources and increases with the increase of the number of CC. the MI simulation analysis is given in Section~\ref{se5-A} to verify the advantage of large channel capacity compared with MIMO-OFDM system.

\subsection{Sensing Performance Analysis} \label{sec4-B}
The CRLB is generally used to determine the lower bound on the variance of unbiased estimator and characterize the sensing ability of a signal \cite{hua2024}. Therefore, the CRLB of the proposed CA-enabled MIMO-OFDM ISAC system is derived in this section.
\subsubsection{CRLB of MIMO-OFDM system}
The CRLB for the OFDM ISAC system has been derived in \cite{wei2023carrier}, and the CRLB for the MIMO-OFDM ISAC system is presented in \hyperref[theorem 3]{\textbf{Theorem 3}}.
For ease of derivation, $\gamma_i=\frac{2v_{i,0}}{c}$ is used as an estimate of velocity to avoid the impact of varying carrier frequencies on theoretical derivation~\cite{wei2023carrier}. 
Meanwhile, $\sin{\theta_{i,\text{Rx}}}$ is used as an estimate of AoA to avoid the presence of unknown parameters after multiple derivations~\cite{Liu2022}.
\begin{theorem}\label{theorem 3}
    For a MIMO-OFDM ISAC system (e.g., the $b$-th CC), the CRLBs of $\sin{\theta_{i,\text{Rx}}}$, $\tau_{i,0}$, and $\gamma_i$ are expressed as (\ref{eq52}), (\ref{eq53}), and (\ref{eq54}), respectively.
    {\fontsize{9}{8} \begin{equation}\label{eq52}
       CRLB(\sin{\theta_{i,\text{Rx}}})= \frac{\Gamma_\text{mimo}N_\text{R}\left[M_bN_b\Theta_b^n\Theta_b^m-\left(\Theta_b^{m,n}\right)^2\right]}{M_bN_b},
    \end{equation}}
    \begin{equation}\label{eq53}   CRLB(\tau_{i,0})=\frac{\Gamma_\text{mimo}N_b\left[M_bN_\text{R}\Theta_b^p\Theta_b^m-\left(\Theta_b^{p,m}\right)^2\right]}{M_bN_\text{R}},
    \end{equation}
    \begin{equation} \label{eq54}  CRLB(\gamma_i)=\frac{\Gamma_\text{mimo}M_b\left[N_bN_\text{R}\Theta_b^p\Theta_b^n-\left(\Theta_b^{p,n}\right)^2\right]}{N_bN_\text{R}},
    \end{equation}
where $\Gamma_\text{mimo}$ is denoted by (\ref{eq55}), 
    \begin{figure*}
     \begin{equation}\label{eq55}
        \Gamma_\text{mimo}= 
         \frac{1}{{\varpi \cdot\left(2\pi\right)^2}}\frac{1}{M_bN_bN_{\text{R}}\Theta_b^p\Theta_b^n\Theta_b^m+2\Theta_b^{p,n}\Theta_b^{p,m}\Theta_b^{m,n}-N_b\left(\Theta_b^{p,m}\right)^2\Theta_b^n-M_b\left(\Theta_b^{p,n}\right)^2\Theta_b^m-N_\text{R}\left(\Theta_b^{m,n}\right)^2\Theta_b^p},
    \end{equation}
  {\noindent} \rule[-10pt]{18cm}{0.1em}
    \end{figure*}
    $\varpi =\frac{\left|A\right|^2}{\sigma^2}$, $\Theta_b^n=\sum\limits_n\left(n\Delta f^b\right)^2$, $\Theta_b^m=\sum\limits_m\left(mf_\mathrm{C}^bT^b\right)^2$, $\Theta_b^p=\sum\limits_p\left(p\frac{d_\text{r}}{\lambda^b}\right)^2$, 
    $\Theta_b^{m,n}=\sum\limits_m
    \sum\limits_n mn\Delta f^bf_\mathrm{C}^bT^b$, $\Theta_b^{p,n}=\sum\limits_p
    \sum\limits_n pn\frac{d_\text{r}}{\lambda^b}\Delta f^b$, and $\Theta_b^{p,m}=\sum\limits_p
    \sum\limits_m pm\frac{d_\text{r}}{\lambda^b}f_\mathrm{C}^bT^b$.
\end{theorem}
\begin{proof}
    Please refer to \hyperref[apdA]{\textbf{Appendix A}}.
\end{proof}

\subsubsection{CRLB of CA-enabled MIMO-OFDM system}
The CRLBs of parameter estimations of the proposed ISAC system are derived.

Before deriving, we need to know that for independent observed data, the total likelihood function is the product of the likelihood functions of each independent observation, and the total log-likelihood function is the sum of the log-likelihood functions of each independent observation. Therefore, for two independent observations, the total Fisher Information Matrix (FIM) is the sum of the FIMs of the two independent observations~\cite{kay1993fundamentals}.

Upon observing (\ref{eq12}), the received echo signals on high and low-frequency bands are independent. Therefore, by linearly summing the respective FIM of the received echo signals on high and low-frequency bands, the CRLB for the CA-enabled MIMO-OFDM system can be derived.

The FIM of the $b$-th CC is expressed as
\begin{equation}\label{eq58}
\mathbf{F}_b=\begin{bmatrix}F_{\sin\theta_{i,\text{Rx}}}^b&F_{\sin\theta_{i,\text{Rx}},\tau_{i,0}}^b&F_{\sin\theta_{i,\text{Rx}},\gamma_i}^b\\F_{\tau_{i,0},\sin\theta_{i,\text{Rx}}}^b&F_{\tau_{i,0}}^b&F_{\tau_{i,0},\gamma_i}^b\\F_{\gamma_i,\sin\theta_\text{Rx}}^b&F_{\gamma_i,\tau_{i,0}}^b&F_{\gamma_i}^b\end{bmatrix},
\end{equation}
where $F_{\sin\theta_{i,\text{Rx}}}^b$, $F_{\tau_{i,0}}^b$, and $F_{\gamma_i}^b$ are expressed in (\ref{eq59}), (\ref{eq60}), and (\ref{eq61}), respectively; $F_{\sin\theta_{i,\text{Rx}},\tau_{i,0}}^b=F_{\tau_{i,0},\sin\theta_{i,\text{Rx}}}^b$ is expressed in (\ref{eq62}); $F_{\sin\theta_{i,\text{Rx}},\gamma_i}^b=F_{\gamma_i,\sin\theta_{i,\text{Rx}}}^b$ is expressed in (\ref{eq63}); $F_{\tau_{i,0},\gamma_i}^b=F_{\gamma_i,\tau_{i,0}}^b$ is expressed in (\ref{eq64}).
\begin{figure*}
\begin{equation}\label{eq59}
F_{\sin\theta_{i,\text{Rx}}}^b=\frac{\left|A\right|^2}{\sigma^2}\sum_p^{N_\text{R}}\sum_m^{M_b}\sum_n^{N_b}\left(2\pi p\frac{d_\text{r}}{\lambda^b}\right)^2=\frac{4\pi^2\left|A\right|^2}{\sigma^2}\left(\frac{d_\text{r}}{\lambda^b}\right)^2\frac{\left(2N_\text{R}-1\right)(N_\text{R}-1)N_\text{R}N_bM_b}{6},
\end{equation}
\begin{equation}\label{eq60}
 F_{\tau_{i,0}}^b=\frac{\left|A\right|^2}{\sigma^2}\sum_p^{N_\text{R}}\sum_m^{M_b}\sum_n^{N_b}\left(2\pi n\Delta f^b\right)^2=\frac{4\pi^2\left|A\right|^2\left(\Delta f^b\right)^2}{\sigma^2}\frac{\left(2N_b-1\right)\left(N_b-1\right)N_bN_\text{R}M_b}6,   
\end{equation}
\begin{equation}\label{eq61}
   F_{{\gamma_{i}}}^{b}=\frac{\left|A\right|^{2}}{\sigma^{2}}\sum_{p}^{{N_{{\text{R}}}}}\sum_{m}^{{M_{b}}}\sum_{n}^{{N_{b}}}\left(2\pi mf_{{\text{C}}}^{b}T^{b}\right)^{2}=\frac{4\pi^{2}\left|A\right|^{2}\left(f_{{\text{C}}}^{b}T^{b}\right)^{2}}{\sigma^{2}}\frac{\left(2M_{b}-1\right)\left(M_{b}-1\right)M_{b}N_{{\text{R}}}N_{b}}{6}, 
\end{equation}
\begin{equation}\label{eq62}
   F_{\sin\theta_{i,\text{Rx}}, \tau_{i,0}}^{b}=F_{\tau_{i,0}, \sin \theta_{i,\text{Rx}}}^{b}=\frac{-|A|^{2}(2 \pi)^{2}}{\sigma^{2}} \sum_{p}^{N_{\text{R}}} \sum_{m}^{M_{b}} \sum_{n}^{N_{b}}\left(p n \frac{d_{\text{r}}}{\lambda^{b}} \Delta f^{b}\right)=\frac{-4 \pi^{2}|A|^{2}\left(\frac{d_{\text{r}}}{\lambda^{b}} \Delta f^{b}\right)}{\sigma^{2}} \frac{\left(N_{\text{R}}-1\right) N_{\text{R}}\left(N_{b}-1\right) N_{b} M_{b}}{4}, 
\end{equation}
\begin{equation}\label{eq63}
F_{\sin \theta_{i,\text{Rx}}, \gamma_i}^{b}=F_{\gamma_i, \sin \theta_{\text{Rx}}}^{b}
=\frac{|A|^{2}(2 \pi)^{2}}{\sigma^{2}} \sum_{p}^{N_{\text{R}}} \sum_{m}^{M_{b}} 
\sum_{n}^{N_{b}}\left(p m \frac{d_{\text{r}}}{\lambda^{b}} f_{\text{C}}^{b} T^{b}\right)
=\frac{4 \pi^{2}|A|^{2}\left(\frac{d_{\text{r}}}{\lambda^{b}} f_{\text{C}}^{b} T^{b}\right)}
{\sigma^{2}} \frac{\left(N_{\text{R}}-1\right) N_{\text{R}}\left(M_{b}-1\right) M_{b} N_{b}}{4},
\end{equation}
\begin{equation}\label{eq64}
F_{\gamma_i, \tau_{i,0}}^{b}=F_{\tau_{i,0}, \gamma_i}^{b}=
\frac{-(2 \pi)^{2}|A|^{2}}{\sigma^{2}} \sum_{p}^{N_{\text{R}}} \sum_{m}^{M_{b}} 
\sum_{n}^{N_{b}}\left(m n \Delta f^{b} f_{\text{C}}^{b} T^{b}\right)
=\frac{-4 \pi^{2}|A|^{2}\left(\Delta f^{b} f_{\text{C}}^{b} T^{b}\right)}
{\sigma^{2}} \frac{\left(M_{b}-1\right) M_{b}\left(N_{b}-1\right) N_{b} N_{\text{R}}}{4}, 
\end{equation}
{\noindent} \rule[-10pt]{18cm}{0.1em}
\end{figure*}

Combining the relationships $\theta_{i,\text{R}}=\text{arcsin}(\sin\theta_{i,\text{R}})$, $\gamma_i=\frac{2v_{i,0}}{c}$, and $\tau_{i,0}=\frac{2r_{i,0}}{c}$, the CRLBs for angle, range, and radial velocity of the CA-enabled MIMO-OFDM system are expressed in (\ref{eq65}), (\ref{eq66}), and (\ref{eq67}), respectively.
\begin{equation}\label{eq65}
CRLB_\text{CA}\left(\theta_{i,\text{Rx}}\right)
=\left[\cos\theta_{i,\text{Rx}}\sum_{b=1}^{\mathcal{B} }\mathbf{F}_{b}\right]_{1,1}^{-1}, 
\end{equation}
\begin{equation}\label{eq66}
CRLB_\text{CA}\left(r_{i,0}\right)=\left[\left(\frac{2}{c}\right)^{2}\sum_{b=1}^{\mathcal{B} }\mathbf{F}_{b}\right]_{2,2}^{-1},
\end{equation}
\begin{equation}\label{eq67}
CRLB_\text{CA}\left(v_{i,0}\right)=\left[\left(\frac2c\right)^2\sum_{b=1}^{\mathcal{B}}\mathbf{F}_1\right]_{3,3}^{-1}.   
\end{equation}

Upon observing (\ref{eq65}), (\ref{eq66}), and (\ref{eq67}), the CRLB of angle is related to angle, which is consistent with previous work in~\cite{Liu2022}. Furthermore, the CRLBs of parameter estimations in CA-enabled MIMO-OFDM system decrease with the increase of the number of CCs, revealing the superiority of the proposed ISAC system. It is important to note that there has still a trade-off between the CRLBs of range and velocity, which can be addressed by combining multi-resource optimization or adaptive techniques~\cite{Liyanaarachchi,liu2018}.

\section{Simulation Results and Analysis} \label{se5}
In this section, the simulation results are provided to verify the high-spend communication and high-accuracy sensing capabilities of the proposed CA-enabled MIMO-OFDM ISAC system. Simulation parameters are detailed in Table \ref{tb1}. The simulation results are obtained with 1000 times Monte Carlo simulations. It is worth noting that reasonable resource allocation can maximize the benefits of communication and sensing, which is also a focus of future trend. To ensure that the simulation results are not constrained by resource allocation, the simulation considers different bandwidths and different numbers of OFDM symbols on high and low-frequency bands. 

\begin{table*}
	\caption{Simulation parameters \cite{wei2023carrier,braun2009parametrization,3gpp2018nr,3gpp.38.104,Mateo}.}
	\label{tb1}
	\renewcommand{\arraystretch}{1.3} 
	\begin{center}\resizebox{\linewidth}{!}{
		\begin{tabular}{|m{0.055\textwidth}<{\centering}| m{0.43\textwidth}<{\centering}| m{0.15\textwidth}<{\centering}| m{0.055\textwidth}<{\centering}|m{0.43\textwidth}<{\centering}| m{0.15\textwidth}<{\centering}|}
			\hline
			\textbf{Symbol} & \textbf{Parameter} & \textbf{Value} & \textbf{Symbol} & \textbf{Parameter} & \textbf{Value} \\
			\hline
			$ f_\mathrm{C}^1 $	& Carrier frequency of low-frequency band &  3.5\;GHz~\cite{3gpp.38.104,Mateo}  & $ f_\mathrm{C}^2 $	& Carrier frequency of high-frequency band &  28\;GHz~\cite{3gpp.38.104,Mateo}  \\
			\hline
			$ M_1$ & Number of OFDM symbols on low-frequency band &  14~\cite{wei2023carrier} & $  M_2 $	& Number of OFDM symbols on high-frequency band & 28 \\
			\hline
			$ N_1$ & Number of subcarriers on low-frequency band &  512~\cite{wei2023carrier}  & $ N_2 $	& Number of subcarriers on high-frequency band &  512 \\
			\hline
              $ T^1 $	& Total symbol duration on low-frequency band  &  43.9\;$\mathrm{\mu}$s   & $ T^2 $	& Total symbol duration on high-frequency band  &  5.49\;$\mathrm{\mu}$s  \\
               \hline
			$ \Delta f^1 $	& Subcarrier spacing on low-frequency band &  30\;kHz~\cite{3gpp2018nr}  & $ \Delta f^2 $	& Subcarrier spacing on high-frequency band & 240\;kHz~\cite{3gpp2018nr} \\
			\hline
			$ N_\text{T} $ & Number of transmit antennas of BS &  128  &   $ N_\text{R} $	& Number of receive antennas of BS &  128  \\			
			\hline
		   $ N_\text{U} $	& Number of receive antennas of UE  & \{4, 5, 6\}  &  $I$  & Number of potential targets & 3  \\ \hline
              $\theta_{i,\text{Rx}}$ & AoAs of targets  & $\{30,40,50\}^\circ$ &  $\theta_{i,\text{Tx}}$ & AoDs of targets  & $\{30,40,50\}^\circ$ \\ 
                \hline
              $ r_{i,0} $	& Ranges of targets  &  \{117,150,170\} m  & $ v_{i,0} $	& Velocities of targets  &  \{13,20,25\} m/s  \\
			\hline
		\end{tabular}}
	\end{center}
\end{table*}

\subsection{Performance of Communication}\label{se5-A}
According to (\ref{eq43}) and (\ref{eq46}) derived in Section \ref{se4-A}, the communication MIs of CA-enabled MIMO-OFDM ISAC system and MIMO-OFDM ISAC system under various SNR and $N_\text{U}$ are simulated in Fig.~\ref{fig5}.
\begin{figure}
    \centering
    \includegraphics[width=0.45\textwidth]{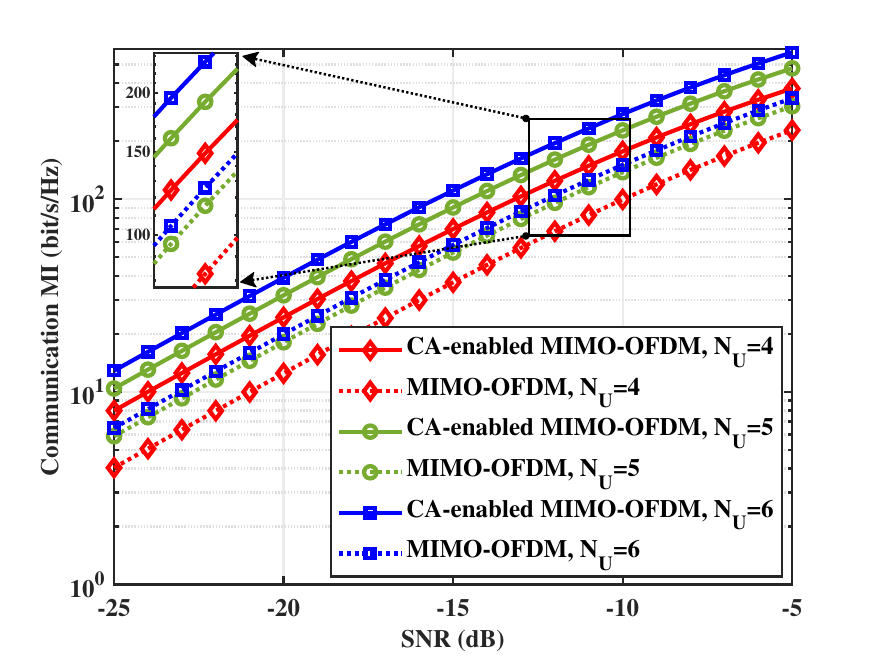}
    \caption{Communication MI of the proposed CA-enabled MIMO-OFDM ISAC system and MIMO-OFDM ISAC system.}
    \label{fig5}
\end{figure}
 
In Fig.~\ref{fig5}, the communication MI of the CA-enabled MIMO-OFDM ISAC signal is larger than that of the MIMO-OFDM ISAC signal under different $N_\text{U}$. Simultaneously, the communication MI rises with the number of receive antennas $N_\text{U}$ in UE, driven by spatial multiplexing in MIMO systems.

\subsection{Performance of Sensing} \label{se5-b}
In this section, the CRLBs and the average of root mean square errors (ARMSEs) of range and velocity estimations are simulated to demonstrate the high-accuracy sensing performance of the proposed CA-enabled MIMO-OFDM ISAC system. 
The ARMSEs are denoted by 
\begin{equation}
\begin{aligned}
    &\Theta_r=\frac{1}{I}\sum_{i=1}^{I}\sqrt{(\hat{r}_{i,0})^2-(r_{i,0})^2}, \\ & \Theta_v=\frac{1}{I}\sum_{i=1}^{I}\sqrt{(\hat{v}_{i,0})^2-(v_{i,0})^2},
\end{aligned} 
\end{equation}
where $\Theta_r$ and $\Theta_v$ denote the ARMSEs of range and velocity estimations, respectively.

\subsubsection{Comparison of CRLB}
According to the closed-form sensing CRLBs derived in Section~\ref{sec4-B}, the CRLBs for range and velocity estimations of CA-enabled MIMO-OFDM ISAC signal, MIMO-OFDM ISAC signal, and OFDM ISAC signal are simulated, as shown in Figs.~\ref{fig6} and \ref{fig7}.
Fig.~\ref{fig6} reveals the following phenomena.
\begin{itemize}
    \item The CRLB of the CA-enabled MIMO-OFDM ISAC signal is the lowest, indicating that the proposed CA-enabled MIMO-OFDM ISAC signal has better range estimation performance than MIMO-OFDM ISAC signal and OFDM ISAC signal.
    \item Comparing with the CRLB of low-frequency MIMO-OFDM ISAC signal, the CRLB of high-frequency MIMO-OFDM ISAC signal is lower, which is caused by the large bandwidth resources of high-frequency signal.
    \item Comparing with the CRLB of OFDM ISAC signal, the CRLB of MIMO-OFDM ISAC signal is lower, which benefits from the spatial diversity of MIMO.
\end{itemize}
Fig.~\ref{fig7} demonstrates the following phenomena.
\begin{itemize}
    \item The CRLB of CA-enabled MIMO-OFDM ISAC signal is lowest, indicating that the CA-enabled MIMO-OFDM ISAC signal has better performance of velocity estimation than MIMO-OFDM ISAC signal and OFDM ISAC signal.
    \item Comparing with the CRLB of low-frequency MIMO-OFDM ISAC signal, the high-frequency MIMO-OFDM ISAC signal is lower, which also exist in the OFDM ISAC signals. The reason is as follows. In the same received slots, the number of OFDM symbols in high-frequency band is more than that in low-frequency band, which brings more velocity information of targets.
    \item Comparing with the CRLB of OFDM ISAC signal, the CRLB of MIMO-OFDM ISAC signal is lower, which also benefits from the spatial diversity of MIMO.
\end{itemize}
Therefore, the simulation results indicate the superior sensing performance of the proposed CA-enabled MIMO-OFDM ISAC signal compared with MIMO-OFDM ISAC signal and OFDM ISAC signal.
\begin{figure}
    \centering
    \includegraphics[width=0.45\textwidth]{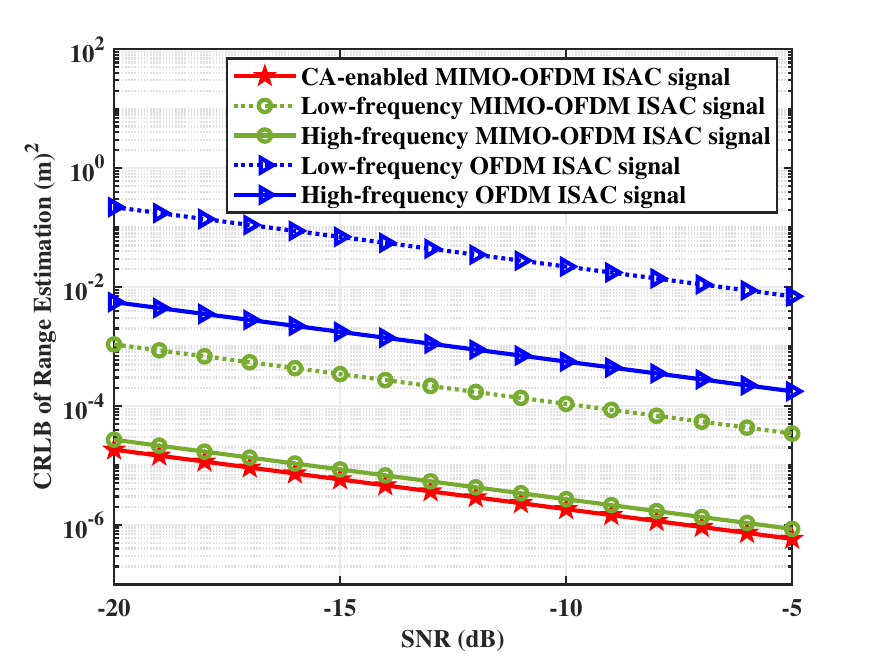}
    \caption{Comparison of CRLB for range estimation}
    \label{fig6}
\end{figure}

\begin{figure}
    \centering
    \includegraphics[width=0.45\textwidth]{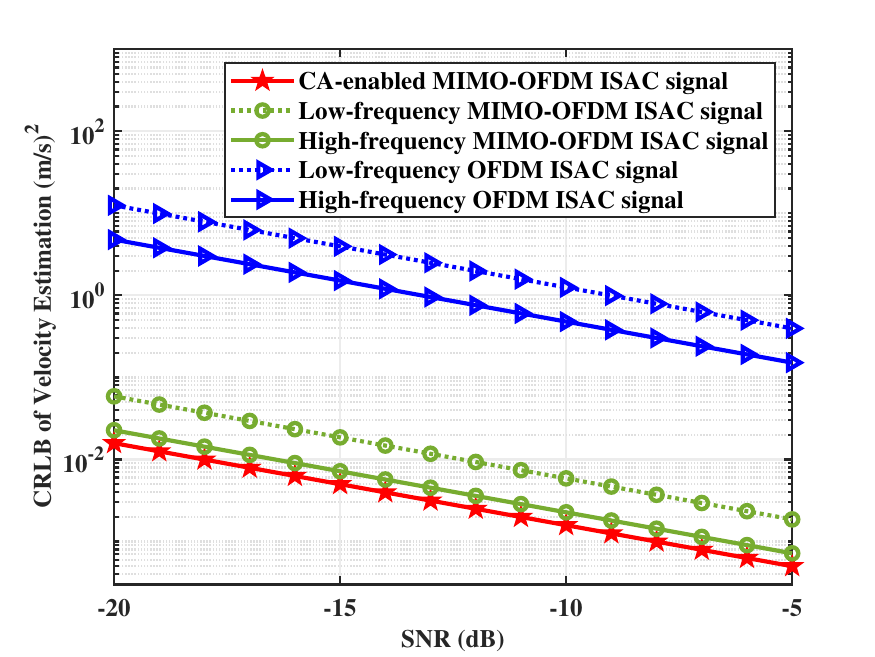}
    \caption{Comparison of CRLB for velocity estimation}
	\label{fig7}
\end{figure}

\subsubsection{Analysis of CRLB}
We further utilize the closed-form CRLB to investigate the relationship between high and low-frequency bandwidth allocation and system sensing performance. Assuming a total available bandwidth of 76.8 MHz, we derive that the subcarrier spacings for high and low-frequency bands satisfy $N_1 + 4\times N_2 = 2560$. The CRLBs for range and velocity estimations are simulated with 
$N_2$ as the variable, as depicted in Fig.~\ref{fig8new}.

As illustrated in Fig.~\ref{fig8new}, the bandwidth allocation is categorized into three cases, where the following conclusions are revealed. 
\begin{itemize}
    \item In Case 1, as the high-frequency bandwidth increases, the CRLB for velocity estimation decreases, whereas the CRLB for range estimation increases. Thus, a trade-off between range and velocity estimation performance emerges.
    \item In Case 2, as the high-frequency bandwidth increases, the CRLB for velocity estimation increases, whereas the CRLB for range estimation first increases and then decreases. Hence, reducing the high-frequency bandwidth optimizes overall sensing performance in Case 2.
    \item In Case 3, as the high-frequency bandwidth increases, the CRLBs for both velocity and range estimations decrease. Therefore, the high-frequency bandwidth should be increased to its maximum available limit.
\end{itemize}
The above conclusions provide guidance for the bandwidth allocation of the proposed CA-enabled MIMO-OFDM ISAC system.

\begin{figure}
    \centering
    \includegraphics[width=0.45\textwidth]{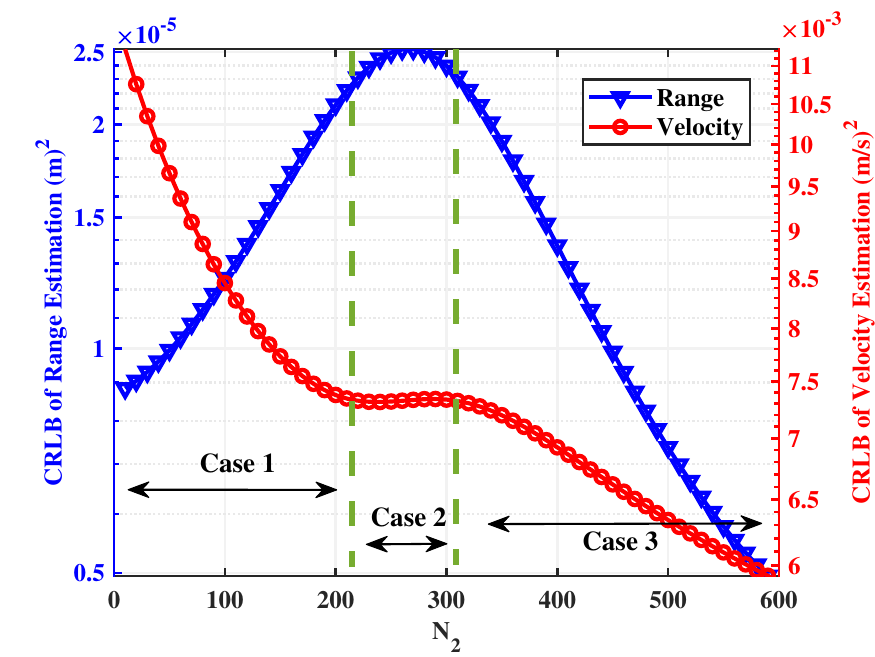}
\caption{CRLBs for different high and low-frequency bandwidth allocations, and the SNR is set as -20 dB}
    \label{fig8new}
\end{figure}

\subsubsection{ARMSE}
The ARMSEs for range and velocity estimations of the proposed CA-enabled MIMO-OFDM ISAC signal and MIMO-OFDM ISAC signal are simulated in Fig.~\ref{fig8}. There are several settings.
\begin{enumerate}
    \item The SNR in the simulation is the transmit SNR of single RE.
    \item The theoretical ARMSEs of range and velocity estimations are expressed in (\ref{ARMSE1}) and (\ref{ARMSE2}), respectively~\cite{skolnik1960theoretical}.
\begin{equation}\label{ARMSE1}
    \sigma_R=\frac{c}{2B\sqrt{2\chi }},
\end{equation}
\begin{equation} \label{ARMSE2}
    \sigma_V=\frac{c}{2M_\text{ofdm}f_\mathrm{C}T_\mathrm{C}\sqrt{2\chi }},
\end{equation}
where $B$ and $\chi $ are the bandwidth and received SNR of signal, respectively; $M_\text{ofdm}$ is the number of OFDM symbols and $f_\mathrm{C}$ is the carrier frequency; $T_\mathrm{C}$ is the total symbol duration.
\end{enumerate}

Fig.~\ref{fig8}(a) reveals the following phenomena.
\begin{itemize}
    \item With the increase of SNR, the ARMSEs of CA-enabled MIMO-OFDM ISAC signal and MIMO-OFDM ISAC signal are stable, because the accuracy of sensing method is limited by the grid size.
    \item Comparing with the ARMSEs of MIMO-OFDM ISAC signals, the ARMSE of CA-enabled MIMO-OFDM ISAC signal is lower, indicating the superior performance of range estimation of the proposed CA-enabled MIMO-OFDM ISAC signal over the MIMO-OFDM ISAC signal. 
    \item When SNR is smaller than $-11.5$ dB, comparing with the ARMSE of low-frequency MIMO-OFDM ISAC signal, the ARMSE of high-frequency MIMO-OFDM ISAC signal is higher. In the simulation, we employed a typical path loss model and set the received SNR of high-frequency signal 5 dB lower than that of low-frequency signal~\cite{goldsmith2005wireless}. Therefore, for comparison under the same received SNR, the curve of high-frequency MIMO-OFDM ISAC signal needs to shift 5 dB to the left.
    \item With the increase of SNR, the ``cliff effect'' is occur in ARMSE value, which is common in signal processing. At low SNR, the sensing information is overwhelmed by noise. When the SNR exceeds a certain threshold, the sensing information stands out. This threshold is also commonly used to characterize the robustness of an algorithm.
\end{itemize}

As shown in Fig.~\ref{fig8}(b), the conclusions are similar to
range estimation. It should be noted that the higher number of OFDM symbols in the high-frequency band results in a lower ARMSE. Additionally, the absence of the ``cliff effect'' is because the threshold does not fall within the SNR range of -20 to -5 dB. 

\begin{figure}
	\centering
	\subfigure[ARMSE for range estimation] {\label{fig8.a}\includegraphics[width=0.45\textwidth]{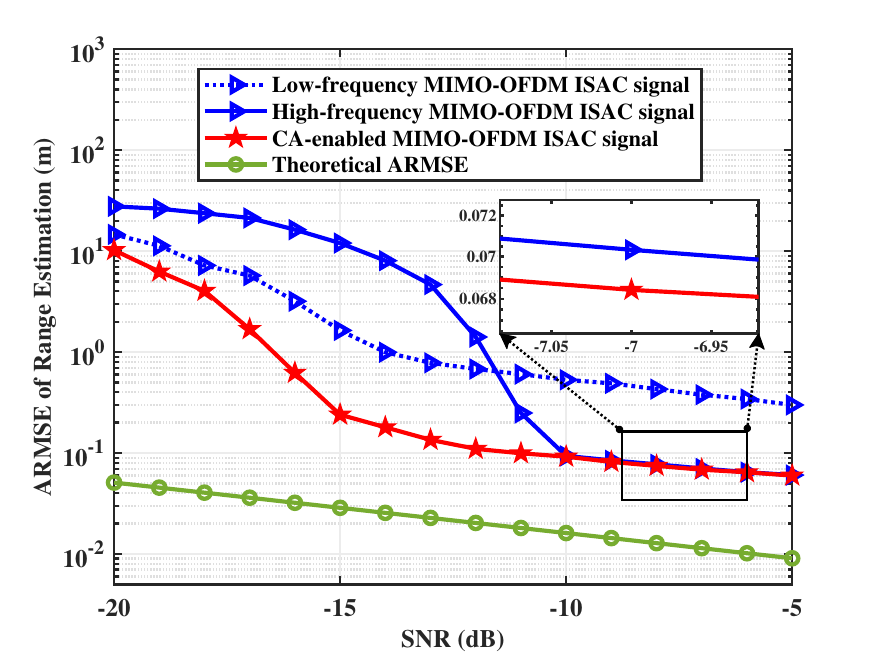}}
	\subfigure[ARMSE for velocity estimation] {\label{fig8.b}\includegraphics[width=0.45\textwidth]{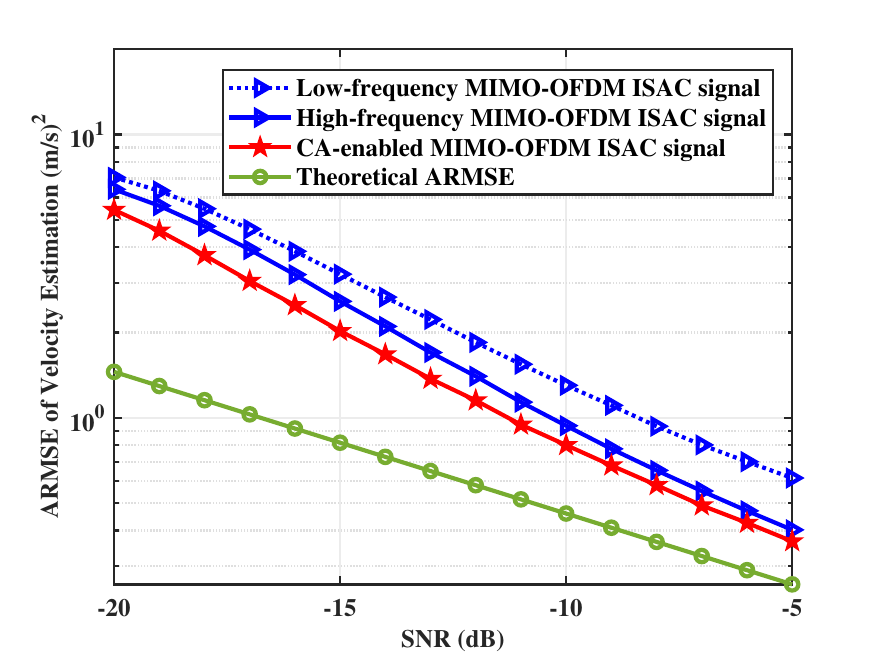}}
	\caption{Comparison of MIMO-OFDM ISAC signal on different frequency bands and CA-enabled MIMO-OFDM ISAC signal.}
	\label{fig8}
\end{figure}

Then, the data-level and symbol-level fusion methods are compared in Fig.~\ref{fig9}, where the following phenomena are revealed.
\begin{itemize}
    \item In Fig.~\ref{fig9}, the ARMSE of symbol-level fusion method is consistently lower than that of the data-level fusion method, indicating the superiority of the proposed symbol-level fusion method over the data-level fusion method.
    \item With the increase of SNR, the ARMSEs of data-level and symbol-level fusions tend to be flat. This is because the accuracy of sensing is not only affected by bandwidth, SNR and other resources, but also restricted by the grid sizes of range and velocity searching intervals. Therefore, when the influence of SNR is dominant, the curve will continue to decline with the increase of SNR, while when the influence of grid is dominant, the curve will be flat.
    \item Compared with the data-level fusion method, the accuracy of range and velocity estimations under symbol-level fusion method can improve by a maximum of 92.77\% and 18.31\%, respectively.
\end{itemize}

\begin{figure}
	\centering
	\subfigure[ARMSE for range estimation] {\label{fig9.a}\includegraphics[width=0.45\textwidth]{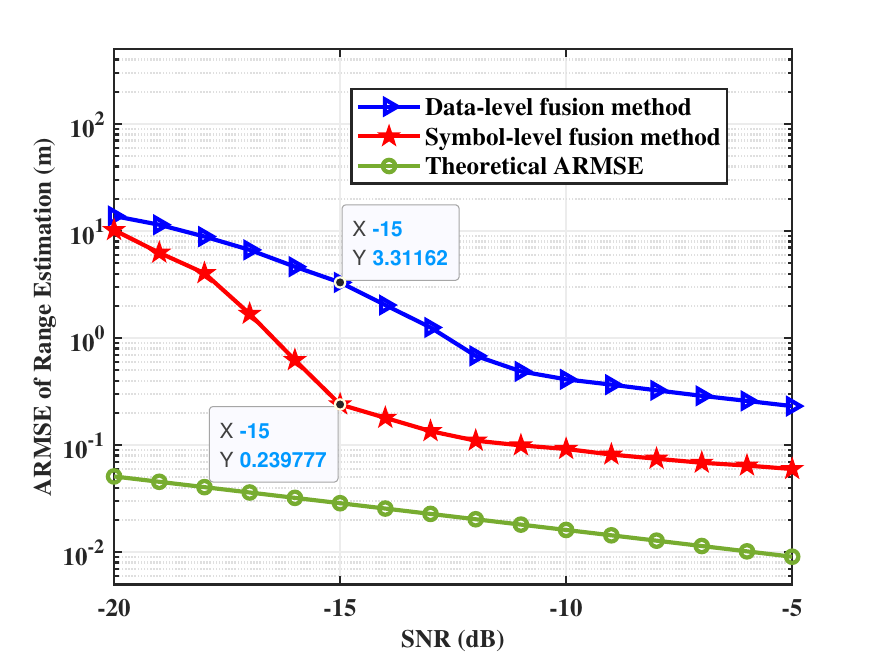}}
	\subfigure[ARMSE for velocity estimation] {\label{fig9.b}\includegraphics[width=0.45\textwidth]{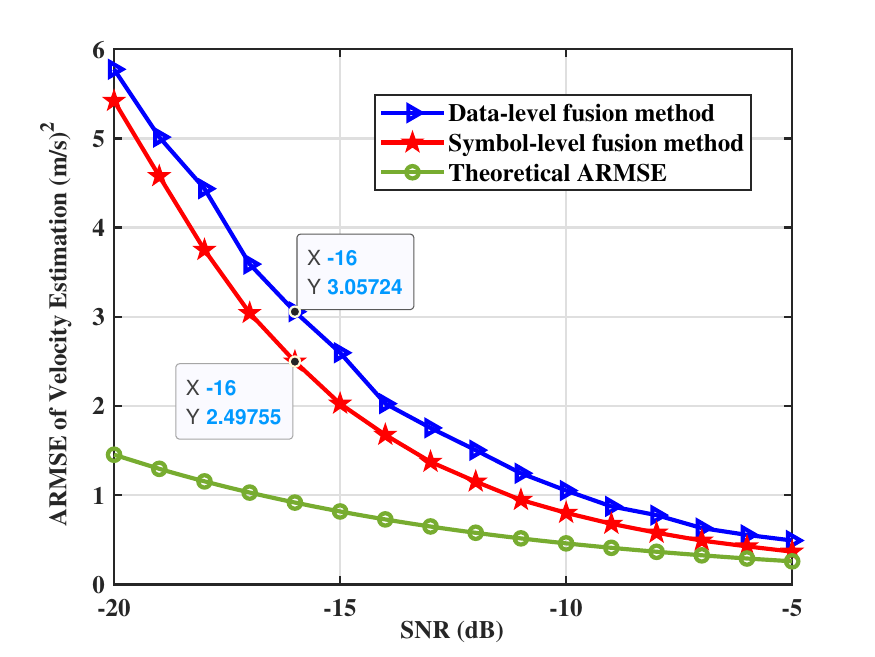}}
	\caption{Comparison of data-level fusion and symbol-level sensing data fusion methods}
	\label{fig9}
\end{figure}

\section{Conclusion} \label{se6}
In this paper, 
CA-enabled MIMO-OFDM ISAC system is investigated, 
where CA aggregates fragmented high and low-frequency bands, achieving the high-speed communication and high-accuracy sensing. 
Firstly, we introduce the signal model of the proposed CA-enabled MIMO-OFDM ISAC system
and the ISAC signal processing framework.
To address the challenges in the sensing signal processing introduced by CA, 
symbol-level fusion method of the sensing data on high and low-frequency bands is proposed. 
In the proposed symbol-level fusion method,
the phase alignment of the echo signals on high and low-frequency bands is realized by the angle compensation, spatial filtering and cyclic cross-correlation operations. Then, the fusion algorithms of feature vectors are proposed to realize high-accuracy sensing.
The simulation results reveal that the performance of symbol-level fusion method is better than that of data-level fusion method.
Furthermore, the closed-form communication MI and sensing CRLB of the proposed ISAC signal are derived.
Numerical results reveal that the proposed CA-enabled MIMO-OFDM ISAC signal exhibits superior communication rate and high-accuracy sensing performance over the conventional MIMO-OFDM ISAC signal. 
This work reveals the potential of CA-enabled MIMO-OFDM ISAC signal in enhancing communication and sensing capabilities.
However, the proposed fusion method does not reach the theoretical lower bound of the proposed CA-enabled MIMO-OFDM ISAC system. Future trends include designing fusion methods and resource allocation methods to approximate the theoretical lower bound.

\begin{appendices}
\section{proof of theorem \ref{theorem 3}} \label{apdA}
According to (\ref{eq12}), in the $b$-th CC, the received echo signal of the $i$-th target by the $p$-th receive antenna on the $n$-th subcarrier during the $m$-th OFDM symbol time is expressed as 
\begin{equation}\label{apeq1}
\begin{aligned}
     y_{p,m,n}^i= & Ae^{j2\pi mf_{\mathrm{C}}^{b}T^{b}\gamma_i}e^{-j2\pi n\Delta f^{b}\tau_{i,0}}e^{j2\pi p\left(\frac{d_{\text{r}}}{\lambda^{b}}\right){\sin\theta_{i,\text{Rx}}}} \\
     & +w_{p,m,n}^i,
\end{aligned}
\end{equation}
where $A \in \mathbb{C}$ and $w_{p,m,n}^i \sim \mathcal{CN}(0,\sigma^2)$ is the AWGN.

The likelihood function for the joint estimation of $\sin\theta_{i,\text{Rx}}$, $\tau_{i,0}$, and $\gamma_i$ is 
\begin{equation}\label{apeq2}
\begin{aligned}
&\ln f\left(y;\sin\theta_{i,\text{Rx}},\tau_{i,0},\gamma_i\right) \\&=  -\frac{M_{b}N_{b}N_{\text{R}}}{2}\ln\left(2\pi\sigma^{2}\right)  \\
&\quad -\frac1{2\sigma^2}\sum_p\sum_m\sum_n\left(y_{p,m,n}^i-s_{p,m,n}^i\right)^*\left(y_{p,m,n}^i-s_{p,m,n}^i\right),
\end{aligned}
\end{equation}
where $y$ is the received echo signal of the $i$-th target, and the PDF is 
\begin{equation}  \label{apeq3} 
\begin{aligned} &f\left(y;\sin\theta_{i,\text{Rx}},\tau_{i,0},\gamma_i\right)\\ &=\frac1{\left(2\pi\sigma^2\right)^{M_bN_bN_\text{R}/2}}e^{-\frac1{2\sigma^2}\sum\limits_p\sum\limits_m\sum\limits_n\left|y_{p,m,n}^i-s_{p,m,n}^i\right|^2},
\end{aligned}
\end{equation} \label{apeq4}
with 
\begin{equation}
    s_{p,m,n}^i=Ae^{j2\pi mf_{\mathrm{C}}^{b}T^{b}\gamma_i}e^{-j2\pi n\Delta f^{b}\tau_{i,0}}e^{j2\pi p\left(\frac{d_{\text{r}}}{\lambda^{b}}\right){\sin\theta_{i,\text{Rx}}}}.
\end{equation}
Based on the relationship between the Fisher information matrix (FIM) and CRLB, we can obtain that 
\begin{equation} \label{apeq5}
{ \fontsize{8}{8}
\begin{aligned}
   & \mathbf{F}^{-1}=\left[
    \begin{array}{ccc}   F_{\sin{\theta_{i,\text{Rx}}},\sin{\theta_{i,\text{Rx}}}}  & F_{\sin{\theta_{i,\text{Rx}}},\tau_{i,0}} &
    F_{\sin{\theta_{i,\text{Rx}}},\gamma_i}\\
    F_{\tau_{i,0},\sin{\theta_{i,\text{Rx}}}}     & F_{\tau_{i,0},\tau_{i,0}} & F_{\tau_{i,0}, \gamma_i}\\
    F_{\gamma_i,\sin{\theta_{\text{Rx}}}} &
    F_{\gamma_i,\tau_{i,0}} & F_{\gamma_i,\gamma_i}
    \end{array}\right]^{-1} \\
    &= \left[
    \begin{array}{ccc}
      CRLB(\sin{\theta_{i,\text{Rx}}}) & Cov(\sin{\theta_{i,\text{Rx}}},\tau_{i,0}) &Cov(\sin{\theta_{i,\text{Rx}}},\gamma_i) \\Cov(\tau_{i,0},\sin{\theta_{i,\text{Rx}}})
     & CRLB(\tau_{i,0}) & Cov(\tau_{i,0},\gamma_i) \\
     Cov(\gamma_i,\sin{\theta_{i,\text{Rx}}}) &
     Cov(\gamma_i,\tau_{i,0}) & CRLB(\gamma_i)
    \end{array}
    \right] ,   
\end{aligned}}
\end{equation}
where $Cov(\cdot)$ is the covariance operation and $F_{\alpha,\beta}=-E\left[\frac{\partial^2\ln f\left(x;\alpha,\beta,...\right)}{\partial\alpha\partial\beta}\right]$.

According to (\ref{apeq2}) and (\ref{apeq5}), the CRLBs of angle, range, and velocity estimations for MIMO-OFDM ISAC system are derived, as shown in \hyperref[theorem 3]{\textbf{Theorem 3}}.
\end{appendices}

\bibliographystyle{IEEEtran}
\bibliography{reference}

\end{document}